\documentclass[1p ]{elsarticle}

\journal{Journal of Economic Theory}
\usepackage{amsfonts}

\usepackage{algpseudocode}
\usepackage{algorithm}

\usepackage{amsmath}  
\usepackage{amssymb}   
\usepackage{amsthm} 

\usepackage{graphicx}
\usepackage{caption}
\usepackage{subcaption}

 \usepackage{hyperref}
\newcommand{\dd}{{\rm d}}

\newcommand{\E}{\mathbb{E}}

\newcommand{\N}{\mathbb{N}}
\newcommand{\R}{\mathbb{R}}
\newcommand{\C}{\mathcal{C}}
\newcommand{\HH}{\mathbb{H}}
\newcommand{\Q}{\mathcal{Q}}
\newcommand{\X}{\mathcal{X}}
\newcommand{\K}{\mathcal{K}}
\newcommand{\Set}{\mathcal{S}}
\newcommand{\CC}{\boldsymbol{C}}
\newcommand{\cc}{c}
\newcommand{\x}{x}
\newcommand{\q}{q}
\newtheorem{problem}{Problem}

\def\dd{{\rm d}}

\newtheorem{theorem}{Theorem}

\newtheorem{lemma}{Lemma}
\newtheorem{corollary}{Corollary}
\newtheorem{definition}{Definition}

\begin{document}

\begin{frontmatter}

\title{Optimal auctions for networked markets with externalities}
 
\author{Benjamin Heymann\footnote{
CMAP, Inria, Ecole polytechnique, CNRS, Universit\'e Paris-Saclay, 91128, Palaiseau, France. \\
CMM, Universidad de Chile, Santiago, Chile.}}

\author{Alejandro Jofr\'e\footnote{CMM and DIM, Universidad de Chile, Santiago, Chile}}

\begin{abstract}
Motivated by the problem of market power in electricity markets, we introduced in previous works a
mechanism  for 
simplified markets of two agents with linear cost.
In standard procurement auctions, the market power resulting from the
quadratic transmission losses allows the producers to bid above their true values, which are their production cost. The mechanism proposed in the previous paper optimally reduces the producers' margin to the society's benefit.
In this paper, we extend those results to a more general market made of a finite
number of agents with piecewise linear cost functions, which makes the
problem more difficult, but simultaneously more realistic. We show
that the methodology works for a large class of externalities.
We also provide an  algorithm to solve the principal allocation
problem.
Our contribution provides a benchmark to assess the sub-optimality of the mechanisms used in practice. 
\end{abstract}

\begin{keyword}
Optimal auctions, mechanism design, allocation algorithm, electricity markets, fixed point.\\
JEL classification: D44, D62, D82
\end{keyword}

\end{frontmatter}

\section{Introduction}
Our purpose in this paper is to show how   oligopolistic
behaviors in network markets can be tackled  using mechanism
design. We point out that the optimal
mechanism we obtain has a surprisingly simple
expression. We complete this work with algorithmic tools for the computation of
this mechanism.
Following a model already discussed in
\cite{escobar2010monopolistic,escobar2008equilibrium,NicolasFigueroaAlejandroJofrBenjaminHeymann},
we consider a geographically extended market where a divisible good is traded. 
In this proposal, each market participant is  located on a node of a graph, and the
nodes are connected by edges.
The good can travel from one node to
another through those edges  at the cost of a  loss. Since our initial motivation was the electricity market, we will do the presentation with quadratic loss, but as explained thereafter, our results extend to a broad class of externalities. 
We are considering the usual transmission network constraints with the DC approximation (active power) for the losses.

We will use the word \textit{principal} to designate what could also be called
in the centralized market literature a central operator, or in the context of electricity
markets, an ISO.
The principal, who aggregates the (inelastic) demand side, has to  locally match -i.e. at each node - production and demand at the lowest
expense through a procurement auction.
As argued in~\cite{NicolasFigueroaAlejandroJofrBenjaminHeymann}, this
setting can be applied to describe real electricity markets, but it could also be used in other markets where a good is being transported.    
Either way, there is a clear antagonism between the market participants: the
operator wants to minimize his expected cost while the producers want to
maximize their expected profits. 
Therefore, at the same time that there is a transaction and a commitment between each agent and the principal, there also exists competition among the agents. 
In a standard procurement auction, the market power resulting
from the quadratic line losses allows the producers to bid above their
true values, or production cost \cite{escobar2010monopolistic}. The
mechanism  reduces the producers' margin and decreases the social
cost represented in this case by the optimal value of the
principal.
This optimal auction design was introduced by Myerson in 1981~\cite{myerson1981optimal} for a non-divisible good and no externalities.

We build on  an electricity market model  introduced by the second author in two previous papers \cite{escobar2008equilibrium} and
\cite{escobar2010monopolistic}. 
The authors wrote a  brief
presentation of this model in \cite{Heymann:2016aa}. 
Other models were proposed for example in  \cite{aussel2013electricity},
\cite{anderson2013mixed}, and \cite{hu2007using}, with a focus  on the existence
of a market equilibrium.
We pinpoint that if our initial motivation was electricity markets, network markets are used in other setting such as telecommunication \cite{Altman2006}.  
Distributed markets were also studied in \cite{Babaioff2009,Cho2003}, with a focus on efficiency and linear cost for transmissions. 
For more information on the techniques we  use  in this paper the reader can refer
to \cite{laffont2009theory}, \cite{krishna2009auction}, \cite{Roughgarden2016},  \cite{nisan2007introduction}, the chapter 45 of \cite{aumann1992handbook}  and 
\cite{topkis1998supermodularity} for general introductions on principal-agent
theory, mechanism design, 
game theory
and lattices theory respectively.

In the sequel we consider, as we did in \cite{NicolasFigueroaAlejandroJofrBenjaminHeymann}, that every participant knows the demand at each node before the interactions begin and that the production cost of each agent is private information. In a standard setting, the agents are the first to bid their costs, after which the principal,
knowing the bids,  minimizes his cost. In a standard setting, the principal is, therefore, a bid-taker. 
The producers know
they influence the allocation and compete with each other to maximize
their individual profit.
Since the demand is known by everyone, everyone can guess the
principal reaction once the bids have been announced: we can therefore virtually remove
the principal   from the interaction in the standard setting and
consider that the agents are  players of  a game with incomplete
information (since the agents do not know their fellow agents'
preferences). This equivalence is true provided that the agents are
not communicating with each other.
The mechanism changes the payoff function of this game -subject to constraints we detail in this article-  so as to minimize the principal's expected cost before the bids are announced.
Allowing the principal to act first by revealing a committing rule
gives him a strategic advantage in the negotiation.

We restrict our discussion to deterministic demand, but the reasoning
extends naturally    to random demand as long as any possible realization of the
demand satisfies the model assumptions. 
Indeed,  since the optimal mechanism constructed in this article is 
incentive compatible, then a random version (where the demand is
revealed after the producers' bidding phase, as in \cite{escobar2008equilibrium}) would be realization-wise
incentive compatible, and so incentive compatible. 
Observe that the mechanism we propose in the sequel could be adapted to elastic, piecewise linear demand.

Our first main result is actually the mechanism design characterization. The result is valid for a very general class of externalities as explained in the generalization section. 
This characterization of the optimal mechanism could be used to  assess the sub-optimality of  the mechanisms used in practice. 

Interestingly, the allocation  procedures for the optimal and the standard mechanism are the same (one just needs to modify the input of the allocation procedure of the standard mechanism to get the allocation of the optimal mechanism).
Our second main result is a principal allocation algorithm based on a fixed point. 
The fixed point could be interpreted as cooperating agents trying to minimize a global criterion by sharing relevant information. 
Our  implementation of the algorithm gives good results against
standard methods.  
We point out that  the numerical computation of the Nash equilibrium
for the procurement auction (important to compare the optimal
mechanism and the standard auction setting) requires an efficient 
algorithm to compute the allocation. 
Some other additional facts are presented within the paper: the
smoothness of the allocation functions ($q$ and $Q$), a decreasing rate
estimation for the fixed point iterations, some results of numerical
experiments with the fixed point algorithm. 

We describe the market in the next section. In \S\ref{sec:mechanism} we introduce and  solve the mechanism design problem. In \S\ref{sec:nagents}, we study the standard allocation problem
and  propose an algorithm to solve it. 
In \S\ref{sec:conclusion} we sum up and comment on our  main results and
propose possible continuations of this work.

\section{Market description}
\label{sec:market}
The production cost of each agent is assumed to be  piecewise linear, non-decreasing and convex  in the quantity produced. 
This class of functions is sufficiently rich to represent real-life
problems and is sufficiently simple for theoretical study.
In this work we need to  assume that the production levels at which
there is a slope change are known in advance and are exogenous - that is the
agents cannot choose them-. 
Then, without loss of generality, we assume that there is a quantity
$\bar{q}$ such that the changes of slope only occur at the multiples
of $\bar{q}$. 
Thus, the authors find it practical to write the production cost functions in the form 
\begin{equation}
\label{cost}
C^{c}(q) = \sum_{j=1}^N c_j \min ( (q-(j-1) \bar{q})^+, \bar{q}),
\end{equation}
where $N\in\N$ and the $c_j$ are some slopes coefficients specific to the agent, while $q$ is the quantity produced. We will sometimes refer to the vector of the $c_j$ as the cost vector (of the agent).
If we denote by $q_i^j$ the quantity produced by agent $i$ at marginal cost $c_i^j$, then  $q_i^j=\min ( (q_i-(j-1) \bar{q})^+, \bar{q})$, where $q_i$ is the total quantity produced by this agent.
Let $c_*<c^*\in \R^{*+}$ and $\CC$ a set of
non-decreasing $N$-tuples of $[c_*,c^*]$.
To each element $c$ of $\CC$ we associate the piecewise linear cost
function $q \rightarrow C^{c}(q) $. Throughout the paper we set, for any $c\in\CC$,
$c^{N+1}=c^*$ to simplify the notations in some proofs. 
Note that in practice a capacity constraint of the type $q\leq j\bar{q}$ for a given agent can be implemented by setting its $(j+1)^{th}$ slope $c_{j+1}$ equal to a big positive number.
If an agent of cost vector $\cc$ produces a quantity $q$ and
receives a transfer $x$, then its profit is
\begin{equation}
  \label{profitDef}
  u_i = x - C^{c}(q).
\end{equation}

There are $n$ agents numbered from $1$ to $n$ in the market.
We denote $I= [1\ldots n]$ and use generically the letter $i$ to refer to a
specific agent, and $-i$ to refer to $I\backslash\{i\}$. We denote $J=
[1\ldots N]$ and we will use
generically $j$ for the cost coefficients of the $jth$ segment
(starting from $1$).
The agents are dispatched on the $n$ nodes of a graph.
At each node $i$ we find the corresponding agent $i$ and a local
demand $d_i$.
The nodes are connected by undirected edges.
We write $V(i)$ the set of nodes different from $i$ connected to $i$.
Obviously if $i_1\in V(i_2)$ then $i_2\in V(i_1)$.
We denote $E = \{(i_1,i_2): i_1\in V(i_2)\}$ the set of undirected edges.
For each $(i_1,i_2)\in E$, we introduce a quadratic loss coefficient $r_{i_1,i_2}$ such that  $r_{i_1,i_2}=r_{i_2,i_1} $.
In the context of electricity markets, this quadratic coefficient
corresponds to the Joule effect within the lines.
We make the non restricting assumption that $N$ is big enough so that in what follows production at each node is smaller than $\bar{q}N$.

We assume that both the agents and the principal are risk neutral:  they maximize their expected profit.  
If the  principal proposes to pay a price $x_i$ to  agent $i$ to
make her produce a quantity $q_i$ - 
this agent being free to accept or decline the offer-
and if the agent $i$ has a production cost defined by $c_i$, then he accepts the offer if
\begin{equation}
\label{IC1}
  x_i - C^{\cc_i}(q_i)\geq 0.
\end{equation}
Then   for agent $i$, either 
$x_i\geq C^{\cc_i}(q_i)$ or 
 $q_i = 0$.
Thus,  if the principal knew the cost vectors $\cc_i$, he would solve an allocation
problem with those $\cc_i$, and then bid to the agents the quantity and
the payments corresponding to the solution of the allocation problem.
But the principal does not know the cost vectors, and instead what
happens is that the
agents  tell him some values for the $\cc_i$ (not necessarily
their real cost vectors), and then the principal decides based on those values.
In this case, previous works \cite{escobar2010monopolistic} showed that the agents
could receive  non-zero profits and  bid above their production costs. 
The question we now address is how to reduce their margins.

To do so, we need to consider an intermediate scenario between the one in which the
agent knows nothing (and is a price taker), and the one in which he
knows everything (and therefore directly optimizes the whole system as a
global optimizer). 
Each agent is characterized by an element $f_i$,  which is a probability density
of support included in  $\CC$ and an element $c_i$  of $\CC$ drawn according to $f_i$. 
Only agent $i$ knows $c_i$, which is private information. The other
agents and the principal only know the probability $f_i$ with which it
was drawn. The density $f_i$ corresponds to the public knowledge on
agent $i$'s production costs so the principal won't accept any bid $c_i$ that is not in the support of $f_i$. 
We assume that the cost slopes are not correlated  for a given agent
and between agents, i.e. their laws $f_i^j$ are independent. In
particular $f_i(c_i) =\prod_{j\in J} f_{i}^j(\cc_{i}^j)$. In such situation, it makes sense to define
\begin{equation}
  \label{eq:1}
f_{-i}(c_{-i}) = \prod_{i'\in I\backslash i} f_{i'}(\cc_{i'})  \quad \mbox{ and }
\quad f(\cc_1,..,\cc_n) = \prod_{i\in I} f_{i'}(\cc_{i}),  
\end{equation}
and $\E$ (respectively $\E_{\cc_{-i}}$) the
mean operator with respect to $f$  (respectively
$f_{-i}$).
The density  $f$ (resp. $f_{-i}$) represents the uncertainty from the
principal's (resp. agent $i$) perspective. 
To simplify notations we will use the symbol $\CC^n$ to denote the
product of the supports of the $f_i$s.
We denote by $\Q$ the set of allocation functions  - which are the applications from $ \CC^n$ to
$\R_+^{n}$, by 
$\X$ the set of payments functions -which are the applications from $ \CC^n$ to
$\R^n$, and by $\HH$ the set of flow functions - which are the applications from  $ \CC^n$ to $\R^E$ -. 
A \emph{direct mechanism}  is a triple $(\q,\x,h)\in (\Q,\X,\HH )$.
Let $(\q,\x)\in (\Q,\X)$. For this payment function and this allocation function,
the expected profit of agent $i$ of type $c_i$ and bid $c_i'$ is 
\begin{equation}
  \label{eq:4}
U_i(\cc_i,c_i')= \E_{-i} u_i =  X_i(c_i')  - \sum_{j\in J} c_i^j Q_i^j(c_i').
\end{equation}
where the capitalized quantities 
\begin{equation}
  \label{eq:2}
  Q_i^{j}(\cc_i)= \E_{-i}  \min ( ( \q_i(\cc_i,\cc_{-i}) -(j-1)
  \bar{q})^+, \bar{q})  \  \mbox{ and }
  \   X_i(\cc_i) = \E_{-i} x_i(\cc_i,\cc_{-i})
\end{equation}
correspond to the average of their  non capitalized counterpart. 
We also denote by
\begin{equation}
  \label{eq:5}
V_i(c_i)= U_i(\cc_i,\cc_i).  
\end{equation}
the expected profit of agent $i$ if he is of type $c_i$ and bids her
true production cost.

For $i\in I$, $j\in J$ and $\cc_i\in\CC_i$ let
$
K^j_i(c_i)= \int_{c_i^{j-}}^{c_i^j} f_i(\cc^{-j}_i,s) ds/f_i(c_i). 
$
We point out that by independence of the laws of the $c_i^j$,
$K^j_i(c_i)= \int_{c_i^{j-}}^{c_i^j}  f_i^j(s) ds \slash
f_i^j(c_i^j)=K_i^j(c_i^j)$.
Thus $K_i^j$ is simply  the ratio of the cumulative distribution and the
probability density for $c_i^j$.
Our main  assumption is  the \emph{discernability assumption}: for all $i\in I$ and $c_i\in\CC_i$, the
virtual cost
$J_{i,j}(c_i^j)=c_i^j+K_i^{j}(c_i^j)$ is increasing in $j$. 
As demonstrated in the next section, the virtual cost could be
interpreted as the real marginal cost augmented by a marginal
information rent.
The assumption imposes the marginal information rent to be such that for any
bid, the virtual marginal
prices are increasing, i.e. the virtual production cost function is
convex.
The assumption is necessary to show the independence property of the
reformulation in  Lemmas \ref{lemma:IC} and \ref{lemma:indep}.

This assumption implies the \emph{non overlapping working zones assumption}:
if we denote by $\CC_i$ the support of $f_i$, then $\CC_i$ should be of
the form:
\begin{equation}
\CC_i = [c_i^{1-},c_i^{1+}] \times\ldots \times[c_i^{N-},c_i^{N+}]
\end{equation}
with $c_i^{1-}<c_i^{1+} <\ldots <c_i^{N-}<c_i^{N+}$.
We could interpret each segment over which the agent has a constant
marginal cost as a working zone with identified productive assets. 
The expertise of the market participants should allow them to, based
on the working zone, assess the marginal cost of the agent.
This makes senses for instance if the setting is repeated over time.
This
estimation  need to be
precise enough so that there is no chance that it corresponds to
another working zone. We use this assumption  in particular in the proof
of lemma \ref{lemma:x}.
For simplicity we assume that
$c_i^j \rightarrow c_i^j+K_j^{i}(c_i^j)$
is increasing in $c_i^j$. 
\footnote{This is  the piecewise linear adaptation of the classic \emph{monotone likelihood ratio property assumption}}. This assumption can be 
withdrawn using the ironing technique introduced by Myerson without difficulty.
To finish with the market presentation, we introduce  the products of
the type sets $\CC^n= \prod_{i\in I}\CC^{i'}$ and  $\CC^{-i}= \prod_{i'\in I\backslash
  \{i\}}\CC^{i'}$.  
\section{Mechanism Design}
\label{sec:mechanism}
We begin with  the revelation principle as expressed in \cite{gibbons1992game}.
\begin{theorem}[Revelation Principle]
To any Bayesian Nash equilibrium of a game of incomplete information, there exists a payoff-equivalent direct revelation mechanism that has an equilibrium where the players truthfully report their types.
\end{theorem}
According to the revelation principle, we can look for direct truthful mechanisms. Because, there is no reason why the agents should willingly report their types  we need to add a constraint on the design to enforce  truthfulness.
This means that the profit of any agent $i$ of type $c_i$ should be maximal when agent $i$ bids her true type $c_i$
i.e. for all  $(\cc'_i,\cc_i)$
\begin{equation}
  U_i(\cc_i,\cc_i) \ge U_i(\cc_i,\cc_i'). \quad (IC)
\end{equation}
This is the incentive compatibility \emph{(IC)} constraint.
In addition, since we want all 
agents to participate in the market, we  need the \emph{participation
constraint} imposing that for all $c_i$
\begin{equation}
   U_i(c_i,c_i) \ge 0.\quad (PC)
\end{equation}
Without this constraint, the principal would optimize as if the agents would accept any deal (even deals where they would make a negative profit). 
The last constraint is that the supply should be at least equal to the
demand at every node. The supply available at a given node is equal
to the production augmented by the imports minus the exports and the
line losses. As explained earlier, there is a loss when some quantity
$h_{i,i'}$ of the divisible good is sent from one node $i$ to another $i'$. This loss is equal to $r_{i,i'}h_{i,i'}^2$, where $ r_{i,i'}$ is a multiplicative constant. 
In order to obtain symmetric expressions, we will proceed as if half
of this quantity was lost by the sender, and the other half by the
receiver (see for instance \cite{escobar2010monopolistic}). 
Note that we could have equivalently used signed flows, but we would
have lost some symmetry in the formulation.
Then the \emph{supply and demand constraint} writes, for all $i \in I$ and $c\in \CC^n$,
\begin{equation}
 q_i(\cc)  + \sum_{i'\in V(i)}h_{i',i}(\cc) - h_{i,i'}(\cc)-\frac{h^2_{i,i'}(\cc)+h^2_{i',i}(\cc)}{2}r_{i,i'} \ge d_i.  \quad (SD)\\
\end{equation} 
We point out that for an optimal allocation (see \S
\ref{sec:nagents}) , $h_{i,i'}h_{i',i}=0$.

The principal decision is a triple $(q,x,h)\in(\Q,\X,\HH)$. This
decision is made under the constraints (IC), (PC) and (SD). Since we
assume that  the principal is risk neutral, his goal is to minimize his average
cost, which translates mathematically by his criterion being equal to the expected  sum of  payments.
Finally the optimal mechanism is the solution of 
\begin{problem}
\label{pb:mechanism}
\begin{equation*} 
\begin{aligned}
& \underset{(\q,\x,h)\in(\Q,\X,\HH)}{\text{minimize}}
\sum_{i\in I} \E x_i(\cc)
 \\
& \text{subject to}\\
& \forall \cc\in\CC^n,\forall i\in I: \quad q_i(\cc)  + \sum_{i'\in
  V(i)}h_{i',i}(\cc) -
h_{i,i'}(\cc)-\frac{h^2_{i,i'}(\cc)+h^2_{i',i}(\cc)}{2}r_{i,i'} \ge
d_i  \  (SD)\\
&\forall \cc\in \CC^n, \forall (i,i')\in E:\quad h_{i,i'}(\cc) \ge 0  \\
 &\forall i\in I,\forall (c'_i,c_i) \in\CC_i^2:\quad  U_i(c_i,c_i) \ge
 U_i(c_i,c_i') \  (IC) \\
&\forall i\in I, \forall c_i\in \CC_i: \quad  U_i(c_i,c_i) \ge 0  \quad (PC).
\end{aligned}
\end{equation*}
\end{problem}
We now proceed to solve the optimal mechanism design problem, which is
a functional optimization problem with an infinity of constraints,
some of which are expressed with integrals. 
The essential observation is that this complicated problem is
equivalent to a much simpler one. The proof relies on the comparison
with two intermediate problems:
\begin{problem}
\label{pb:modified}
\begin{equation*} 
\begin{aligned}
& \underset{(\q,\x,h)\in(\Q,\X,\HH)}{\text{minimize}} \sum_{i\in I} \E 
x_i(\cc) 
 \\
& \text{subject to.}\\
& \forall \cc\in\CC^n,\forall i\in I: \quad q_i(\cc)  + \sum_{i'\in V(i)}
h_{i',i}(\cc) - h_{i,i'}(\cc) - \frac{h^2_{i,i'}(\cc)+h^2_{i',i}(\cc)}{2}r_{i,i'} \ge d_i (SD)  \\
&\forall \cc\in \CC^n, \forall (i,i')\in E:\quad h_{i,i'}(\cc) \ge 0  \\
&\forall i\in I,\forall j\in J, (\cc^{-j},t_1,t_2), 
 (c^1,\ldots,t_k,\ldots , c^N) \in\CC_i,:  V_i(c^1,..,c^{j-1},t_1,c^{j+1}..,c^N)  \\
&-V_i(c^1,..,c^{j-1},t_2,c^{j+1}..,c^N) =
\int^{t_2}_{t_1}  Q_i^j(c^1,..,c^{j-1},s,c^{j+1}..,c^N) \dd s
\quad (H1)\\
& \forall i\in I, \forall (\cc,\cc')\in \cc^2: \quad (\cc - \cc').(Q_i(\cc)-
    Q_i(\cc') )\leq 0, \quad \quad (H2)\\
&\forall i\in I, \forall c_i\in \CC_i: \quad  V_i(c_i) \ge 0  \quad (PC),
\end{aligned}
\end{equation*}
\end{problem}
and  
\begin{problem}
\label{pb:reformulation}
\begin{equation*} 
\begin{aligned}
&\underset{(\q,h)\in(\Q,\HH)}{\text{minimize}}
\E  \sum_{i\in I}  \sum_{j\in J}  q_i^j (c)(c_i^j+ K^j_i(c_i^j))\\
& \text{subject to}\\ 
&\forall (\cc,i)\in \CC^n \times I: q_i(\cc)  + \sum_{i'\in V(i)}
h_{i',i}(\cc) - h_{i,i'}(\cc) - \frac{h^2_{i,i'}(\cc)+h^2_{i',i}(\cc)}{2}r_{i,i'} \ge d_i   (SD) \\
&\forall \cc\in \CC^n, \forall (i,i')\in E:\quad h_{i,i'}(\cc) \ge
0.  \\
&\forall \cc\in\CC_i,\forall i\in I: x_i(\cc)= \sum_{j\in
  J}q_i^j(\cc)c_i^j+
\int_{c_i^j}^{c_i^{j+}}q_i^j(c_i^1\ldots
c_i^{j-1},t,c^{(j+1)+}_1\ldots c^{N+}_i;c_{-i})\dd
t.
\end{aligned}
\end{equation*}
\end{problem}
The inequality on the scalar product in (H2) is the piecewise linear equivalent of a
monotonicity condition already encountered in
\cite{NicolasFigueroaAlejandroJofrBenjaminHeymann}.  
The first two problems are very similar, but (IC) has been replaced by
(H1) and (H2) and (PC) is expressed in terms of $V$ instead of $U$. This replacement is a trick introduced by Myerson in his 1981 paper.
We will show later on how we can compare Problems 2 and 3, but note that 
Problem 3 is  simpler, as the optimization part can be solved
pointwise (and $x$ can be deduced from this pointwise optimization). 
The main result of this paper is that  the three problems have the same solution. 
\subsection{Necessary conditions for Problem \ref{pb:mechanism}}
\label{sec:equivalence of the problem}
We  derive some necessary conditions for a solution of  Problem
\ref{pb:mechanism}.
In fact, we only use constraint $(IC)$ to deduce the two next results.
The first lemma indicates that any solution of the first problem
should be such that $Q$ is monotonous. This is a classic result
already introduced  in  \cite{myerson1981optimal} and
\cite{NicolasFigueroaAlejandroJofrBenjaminHeymann}, for instance. The novelty 
here is that in the context of  piecewise linear production cost functions,  this monotonicity result is expressed  in a vectorial sense.
\begin{lemma}[$Q$ monotonicity]
\label{lemma:monotone}
  If $(q,x,h)$ is admissible for Problem \ref{pb:mechanism}, then for all agent $i\in I$ and all  $( \cc_i,\cc'_i )\in \CC_i^2$
  \begin{equation} ( \cc_i - \cc_i'
    ).(Q_i(c_i)-
    Q_i(\cc'_i ) )\leq 0
  \end{equation}
where $.$ is the scalar product in $\R^N$.
\end{lemma}
\begin{proof}
We omit the $i$ in the proof, as it plays no role.
First, let $(\cc,\cc')\in \CC_i^2$
by the (IC) constraint,
\begin{equation} 
U(\cc,\cc)\geq U(\cc,\cc') \quad \mbox{and} \quad 
U(\cc',\cc')\geq U(\cc',\cc)
\end{equation}
i.e.
\begin{equation} 
\begin{aligned}
  X(c)  - \sum_{j\in J} c^j Q^j(c) \geq   X(c')
  - \sum_{j\in J} c^j Q^j(c')\\
 X(c')  - \sum_{j\in J} c^{j'} Q^j(c') \geq   X(c)
  - \sum_{j\in J} c^{j'} Q^j(c) .
\end{aligned}
\end{equation}
We get the lemma after the summation of the two inequalities and simplification.
\end{proof}

Lemma \ref{lemma:monotone} indicates that an agent should be 
producing less on average in his $i$th working zone if he is bidding a higher
marginal cost for this working zone.
\begin{lemma}
\label{lemma:int}
  If $(q,x,h)$ is admissible for Problem \ref{pb:mechanism} then for any agent (omitting $i$) for
  any $\cc$, $t_1$ and $t_2$
\begin{equation}
\begin{aligned}
V(c^1,\ldots,c^{j-1},t_1,c^{j+1},\ldots,c^N) =&V
(c^1,\ldots,c^{j-1},t_2,c^{j+1},\ldots,c^N) \\
 &-\int^{t_1}_{t_2}  Q^j(c^1,\ldots,c^{j-1},s,c^{j+1},\ldots,c^N)\dd s
\end{aligned}
\end{equation}
\end{lemma}
\begin{proof}
The inequality $
U(c,c)\leq U(c,c')$ implies that $c' \rightarrow
U(c,c')$ is maximal at $c$ for any $c\in\C_i$.
Moreover, \begin{equation}
t \rightarrow U((c^1,..,c^{j-1},t,c^{j+1}..,c^N),\cc)=
X(\cc)  - \sum_{k\in J \backslash \{ j\}}
c^k Q^k(\cc) -
tQ^j(\cc)
\end{equation}
 is absolutely continuous, differentiable with respect to $t$ for all
$c$, and its derivative is $-Q^j(c)$. By definition of $q^j$, 
$Q^j\leq \bar{q}$.
The envelope theorem yield the result. 
\end{proof}

\subsection{Necessary conditions for Problem \ref{pb:modified}}
We  derive some necessary conditions for a solution of  Problem \ref{pb:modified}.
\begin{lemma}
\label{lemma:V}
 If $(q,x,h)$ is an optimal  solution to Problem \ref{pb:modified} then
 (omitting $i$) for all $c\in \CC_i$
\begin{equation}
V(c) = \sum_{j\in J} \int_{c^j}^{c^{j+}}
Q^j(c^1\ldots c^{j-1},t,c^{(j+1)+},\ldots, c^{N+})\dd t.
\end{equation}
\end{lemma}
\begin{proof}
According to (H1) 
\begin{eqnarray*}
 \sum_{j\in J}
  \int_{c_j}^{c^{j+}}Q^j(c^1\ldots c^{j-1},t,c^{(j+1)+},\ldots,
  c^{N+})\dd t=\\ \sum_{j\in J}
  V(c^1,..,c^{j-1},c^j,c^{(j+1)+},\ldots, c^{N+})
  -V(c^1,..,c^{j-1},c^{(j)+},\ldots, c^{N+}) \\
=V(\cc) - V(c^{1+},\ldots, c^{N+}). 
\end{eqnarray*}
This is an expression for $V(\cc)$ as a sum of a positive function of $\cc$ and a
constant $V(c^{1+},\ldots, c^{N+})$. It is clear that to optimize the criteria, this constant
should  be  as small as possible. The participation contraint (PC)
imposes that $V(c^{1+},\ldots, c^{N+}) \geq 0$, therefore $V(c^{1+},\ldots, c^{N+}) = 0$.
\end{proof}

A consequence of this is: 
\begin{corollary}
\label{lemma:v0}
 If $(q,x,h)$ is  an optimal  solution of Problem \ref{pb:modified}
 then for all $i\in I$, 
\begin{equation}
V_i(c^{1+}_i,\ldots, c^{N+}_i) = 0.
\end{equation}
\end{corollary}
\begin{proof}
See the proof of Lemma \ref{lemma:V}.
\end{proof}

Corollary \ref{lemma:v0} means that  if an agent bids a production cost
function that is the maximum of what he could bid, he should not make
any profit, which is why he should be paid exactly his production cost. We
see with this lemma that if the public information  is
inaccurate and the real cost
of an agent is higher than what could be expected, then there is a
risk that the participation constraint is not satisfied. 
On the other hand, it should not be surprising that an 
agent can have a zero profit: remember that in the extreme case in which the principal
knows everything (discussed in \S\ref{sec:market}), the agents do not make any profit.

Another consequence of lemma \ref{lemma:V} is
\begin{lemma}
\label{lemma:x}
 If $(q,x,h)$ is an optimal  solution of Problem \ref{pb:modified},
 the expected profit of agent $i$  (over his type) is
\begin{equation}
\E V_i(\cc) = \sum_{j\in J} \int_{(\cc_1..\cc_n) \in \CC_i } Q^j_i
(c^1,\ldots,c^j,c^{(j+1)+},\ldots c^{N+}) K^j_i(c)f_i(c)\dd c.
\end{equation}
\end{lemma}
\begin{proof}

By Lemma \ref{lemma:V} and Fubini's lemma, $\E V_i(c)$ is equal to 
\begin{eqnarray*}
\E \sum_{j\in J } \int_{c^j}^{c^{j+}}
Q_i^j(c^1,\ldots,c^{j-1},t,c^{(j+1)+},\ldots c^{N+})\dd t\\
=
\sum_{j\in J }  \int_{c^{-j}\in\CC^{-j}} \int_{c^j = c^{j-}}^{c^{j+}} \int_{t=c^j}^{c^{j+}}
Q_i^j(c^1,\ldots,c^{j-1},t,c^{(j+1)+}_i,\ldots c^{N+}_i)f_i(c)    \dd t
  \dd c^j  \dd c^{-j}.
\end{eqnarray*}
Our task is now to compute the inner term.  
Applying again Fubini's lemma, this term is equal to 
\begin{eqnarray*}
\int_{c^j = c^{j-}}^{c^{j+}} \int_{t=c^j}^{c^{j+}}
Q_i^j(c^1,\ldots,c^{j-1},t,c^{(j+1)+},\ldots c^{N+})f_i(c) \dd t \dd
  c^j=\\
\int_{t = c^{j-}}^{c^{j+} } \int_{c^j=c^{j-}}^{t}
Q_i^j(c^1,\ldots,c^{j-1},t,c^{(j+1)+},\ldots c^{N+})f_i(c)   \dd c^j  \dd t=\\
\int_{t = c^{j-}}^{c^{j+}} 
Q_i^j(c^1,\ldots,c^{j-1},t,c^{(j+1)+},\ldots c^{N+}) (\int_{c^j=c^{j-}}^{t}
  f_i(c)   \dd c^j)  \dd t=\\
\int_{t = c^{j-}}^{c^{j+}}  Q_i^j(c^1,\ldots,c^{j-1},t,c^{(j+1)+},\ldots c^{N+})
  (\int_{c^j=c^{j-}}^{ t}
 \frac{f_i(c)} {f_i(c^{-j}, t)} \dd c^j)
  f_i(c^{-j}, t)   \dd t =\\
\int_{t = c^{j-} }^{c^{j+} }  Q_i^j(c^1,\ldots ,c^{j-1},t,c^{(j+1)+},\ldots c^{N+})K_i^j(t) f_i(c^{-j}, t)   \dd t =\\
\int_{c^j= c^{j-}}^{c^{j+}}  Q_i^j(c^1,\ldots,c^{j-1},c^j,c^{(j+1)+},\ldots c^{N+})
 K_i^j(c^j)
  f_i(c_i)   \dd c^j 
\end{eqnarray*}
We get the lemma  by summing all the inner terms.
\end{proof}

\begin{lemma}
\label{lemma:H1}
If  (H1) is satisfied, then for any  $(a,b)\in \CC_i^2$ (omitting $i$) 
\begin{equation}
\label{eq:xaxb}
 X(a) - X(b) = \sum_{j\in J} [ a^j Q^j(a) - b^j
 Q^j (b) + \int_{a^j}^{b^j}Q^j
 (b^1\ldots b^{j-1},t,a^{j+1} \ldots a^N)\dd t ]
\end{equation}
\end{lemma}
\begin{proof}
Because of its length the proof is detailed in Appendix \ref{appendix:lemmaH1}
\end{proof}
\begin{lemma}
\label{lemma:IC}
If $(q,x,h)$ verifies  (H1) and (H2) and $Q_i^j$ is independent of
$c_i^{j'}$ for $j'> j$, then for all $(\cc, \tilde{\cc}) \in \CC^2$
\begin{equation}
U(\cc,\cc) \geq U(\cc,\tilde{\cc}).
\end{equation}
\end{lemma}
\begin{proof}
Since (H1) is satisfied, equation \eqref{eq:xaxb} of Lemma
\ref{lemma:H1} applies.
We combine this relation with the definition of the expected profit $U$ from
\eqref{eq:4}. We obtain:
\begin{multline*}
U(c,c)- U(c,\tilde{c})  
 =\sum_{j\in J} c^j Q^j(c) -  \tilde{c}^j Q^j(\tilde{c}) + \\
 \int_{c^j}^{\tilde{c}^j}
 Q^j(\tilde{c}^1,...,\tilde{c}^{j-1},t,c^{j+1},...c^N) \dd t
+  c^j Q^{j}(\tilde{\cc}) - c^j Q^{j}(\cc )\\
 =\sum_{j\in J}    (c^j -  \tilde{c}^j )Q^j(\tilde{c}^1,...,\tilde{c}^{j-1},\tilde{c}^j)) +
 \int_{c^j}^{\tilde{c}^j}Q^j(\tilde{c}^1,...,\tilde{c}^{j-1},t)\dd t\\
 =\sum_{j\in J}   \int_{c^j}^{\tilde{c}^j}
 Q^j(\tilde{c}^1,...,\tilde{c}^{j-1},t)-Q^j(\tilde{c}^1,...,\tilde{c}^{j-1},\tilde{c}^j)\dd t,
\end{multline*}
where we used the independence hypothesis for the second equality.
By
(H2), which implies the decreasingness of $Q^j$ with respect to
$c_i^j$ when all other quantities are fixed, if $c^j <\tilde{c}^j$ then for  any $t\in [c^j,\tilde{c}^j]$,  $Q^j(t)-Q^j(\tilde{c}^j)\geq 0 $. Otherwise,
we use  the formula $\int_a^b = -\int_b^a$ and the fact that
 any $t\in [\tilde{c}^j,c^j]$ verifies $Q^j(t)-Q^j(\tilde{c}^j)\leq 0
 $.
Therefore $U(c,c)- U(c,\tilde{c})$ is non negative. 
\end{proof}
\subsection{Necessary conditions for Problem \ref{pb:reformulation}}
We derive some properties for Problem
\ref{pb:reformulation}.
\begin{lemma}
\label{lemma:indep}
There is an optimal solution  $(q,x,h)$ for Problem
\ref{pb:reformulation} such that   $q_i^j$
(and $Q_i^j$) is
independent of $c_i^k$ for $k\neq j$.
\end{lemma}
\begin{proof}
First note that $x$ is not taking any role in the optimization
problem: it is defined afterward. The only real optimization variables
are then $q$ and $h$. Remember that $q_i^j$ is defined as a function
of $q$ by  $q_i^j=\min ( (q_i-(j-1) \bar{q})^+, \bar{q})$.
The constraints are defined for each $c\in \CC^n$ and the integral criterion is in fact a
sum of independent criteria depending on $q(c)$ for
$c\in\CC^n$. Therefore we can solve  Problem 3 with a pointwise optimization.
By the \emph{discernability assumption}, for any $c\in \CC^n$ and
$i\in I$,
$c_i^j+K_i^j(c_i^j)$ is increasing in $j$. Therefore for all $c \in \CC^n$,
$i\in I$, $
\sum_{j\in J}  q_i^j (c)(c_i^j+ K^j_i(c_i^j))$ is a convex criteria in
$q_i$ and therefore the pointwise problem corresponds to Problem
\ref{pb:alloc} of \S\ref{sec:nagents}. In particular, we can apply Lemma~\ref{lemma:indep1} from the
next section. Thus $q_i^j$ only depends on $c_i^j$ and $c_{-i}$. This
property is preserved by integration over the $c_{-i}$: $Q_i^j$ only
depends on $c_i^j$.  
\end{proof}

We point out that,  since the pointwise problem has a unique solution,
the pointwise optimal solution introduced in the proof is uniquely defined.
\begin{theorem}
\label{th:Qc1}
If $(q,x,h)$ is the pointwise optimal solution of Problem
\ref{pb:reformulation} and $K_i^j$ is smooth in $c_i^j$ for
$(i,j)\in I \times J$ and $c\in \CC_i$, then for all $i\in I$,  $Q_i$
is $C^\infty$ over $\CC_i$.
\end{theorem}
\begin{proof}
Remember that $c_i^j \rightarrow c_i^j + K_i^j(c_i^j)$ is increasing,
so by composition with smooth bijection,  we
can do the proof as if the costs involved were $c_i^j$ instead of $c_i^j + K_i^j(c_i^j)$.
According to Lemma \ref{lemma:qc1}, $q_i$ is continuous. 
Since $q_i$ is bounded, we can apply the dominated convergence theorem
to show that $Q_i$ is continuous.
We can then  we proceed by mathematical induction.
Assume that $Q_i$ is $C^l$, then take $c_i^0\in \CC_i$ and
$c_i^k $ a sequence in $\CC_i$ that converges to $c_i^0$. 
Since $\hat{\Set} = \cup_{k\in\N} \Set(c_i^k)$ is  a countable union
of null measured set (by Lemma \ref{lemma:nullMeasure2}), its measure
is zero. Without changing the results, we can compute the integrals on
$\CC^{-i} \backslash \hat{\Set}$ instead of $\CC^{-i}$. Since $q_i$
and its derivatives are bounded, we can apply the dominated
convergence theorem to compute the limit of $\frac{Q_i^{(l)}(c_i^0)-
  Q_i^{(l)}(c_i^k)}{c_i^0-c_i^k}$ as $k$ goes to $+\infty$ as the
integral of a limit.
Since we removed the point over which this limit was not defined, we
get that  $\frac{Q_i^{(l)}(c_i^0)-
  Q_i^{(l)}(c_i^k)}{c_i^0-c_i^k}$ has a limit, and this limit does not
depend on the sequence $c_i^k$. So $Q_i$ is $l+1$ times derivable at $c_i$, for
all $c_i$. 
We conclude by induction.
\end{proof}

\subsection{Resolution of the mechanism design problem}
Last but not least, we state the main result of the Section.
\begin{theorem}
\label{th:mainResult}
Let $(q_i^j,h)$ be defined such that for any $c\in\CC^n$,
$(q_i^j(c_i^j,c_{-i}),h(c))$ solves
\begin{equation*} 
\begin{aligned}
&\underset{\q_i^j,h}{\text{minimize}}
\sum_{i\in I}  \sum_{j\in J}  q_i^j (c_i^j,c_{-i})(c_i^j+ K^j_i(c_i^j))\\
& \text{subject to}\\ 
&0\le q_i^j\le \bar{q}\\
&\sum_{j\in J}  q_i^i(\cc_i^j,c_{-i})  + \sum_{i'\in V(i)}
h_{i',i}(\cc) - h_{i,i'}(\cc) - \frac{h^2_{i,i'}(\cc)+h^2_{i',i}(\cc)}{2}r_{i,i'} \ge d_i   \\
& h_{i,i'}(\cc) \ge
0,
\end{aligned}
\end{equation*}
and set 
\begin{equation}
q_i (c)= \sum_{j\in J}q_i^j(\cc_i^j,c_{-i}) \mbox{ and } 
x_i(\cc)= \sum_{j\in
  J}q_i^j(\cc_i^j,c_{-i})c_i^j+
\int_{c_i^j}^{c_i^{j+}}q_i^j(t,c_{-i})\dd
t,
\end{equation}
then $(q,h,x)$ solves the optimal mechanism design problem (Problem 1). 
\end{theorem}
\begin{proof}
\begin{itemize}
\item First note that $(q,h,x)$ is the pointwise solution of Problem 3 so it is
  optimal for Problem 3, moreover, by construction $(q,h,x)$ satisfies
  (SD) and $h\geq 0$.
\item Then note that by Lemma \ref{lemma:x}, $(q,h,x)$ solves a
  relaxation of Problem 2, but is it admissible for Problem 2 ?
\item By definition of $V$ (omitting $i$), 
\begin{eqnarray*}
V(c_1 \ldots a_j\ldots  c_N) - V(c_1 \ldots b_j\ldots c_N)=\\\E x(c_1
  \ldots a_j\ldots  c_N) - x(c_1 \ldots a_j\ldots  c_N)
  -[Q^j(a^j)a^j - Q^j(b^j)b^j]= \\
\E q_i^j(a^j,c_{-i})a^j+
\int_{a^j}^{c_i^{j+}}q_i^j(t,c_{-i})\dd
t- \E q_i^j(b^j,c_{-i})b^j-
\int_{b^j}^{c_i^{j+}}q_i^j(t,c_{-i})\dd
t     \\
-[Q^j(a^j)a^j - Q^j(b^j)b^j]=
\E\int_{a^j}^{b^j}q_i^j(t,c_{-i})\dd
t=     \int_{a^j}^{b^j}Q_i^j(t)\dd
t
\end{eqnarray*}
where we used the definition of $x$, the definition of $Q$ and Fubini
lemma's for the second, third and fourth equalities. 
Threfore $(q,h,x)$ satisfies (H1).
\item By construction, $q_i^j$ is non-increasing in $c_i^j+
  K_i^j(c_i^j)$ then using the third assumption, $q_i^j$ is non-increasing in $c_i^j$ so for any
   $(a,b,c_{-i})\in \CC^2\times \CC^{-i}$, $(a_i^j
  - b_i^j)(q_i^j(a_i^j,c_{-i})-q_i^j(b_i^j,c_{-i}))\leq 0$, so  by
  integration with respect to $c_{-i}$, $(a_i^j
  - b_i^j)(Q_i^j(a_i^j)-Q_i^j(b_i^j)\leq 0$ and then by summation over
  $j$, 
$(\cc - \cc').(Q_i(\cc)-
    Q_i(\cc') )\leq 0$, i.e. (H2) is satisfied.
\item Since (H1)  is satisfied, $V_i(c_i)\geq V_i(c_i^+) $. Moreover, $V_i(c_i^+) =0$ by
  construction of $x$. Therefore the participation constraint (PC) is satisfied.
\item Therefore $(q,h,x)$ is admissible for Problem 2. So it solves
  Problem 2.
\item Since $(q,h,x)$ solves Problem 2, by Lemma \ref{lemma:IC} the
  incentive compatibility constraint (IC) is satisfied. Moreover, by
  Lemma \ref{lemma:V}, (PC) is satisfied. Thus $(q,h,x)$ is admissible
  for Problem 1, but is it optimal ?
\item By Lemmas \ref{lemma:monotone} and \ref{lemma:int}, any optimal
  solution of Problem 1 should be admissible for Problem 2. Since the
  criteria are the same, we conclude that $(q,h,x)$ is an optimal
  solution of Problem 1. 
\end{itemize}
\end{proof}

\subsection{Comments}

In the optimal mechanism, the agents are paid at a marginal
price which is equal to their bid augmented by an information rent.
This information rent depends on the problem structure since  it is built from a collection of  allocation problems, and it depends on the
available information by the fact that, in these optimization problems,
the marginal prices are replaced by the virtual marginal prices
$c_i^j+K_i^j(c_i^j)$. 
We point out that, as already noted  in
\cite{correa2009planner}, the computation of such rent may  pose  a
practical difficulty for large problems.

Notice that, by construction, the optimal  mechanism is incentive
compatible  no
matter  the value of $K$ 
since (H1) is verified anyway as long as the hypotheses are satisfied. 
If this market is repeated over time, the principal can learn the distribution of the producers' cost parameters.

The model extends to the case in which  some nodes do not
have a producer and where for others, the demand is null. 
In particular, we can consider the buyer/suppliers setting where there is demand
only at one node.   

One may argue that one limit of the current result is that it does not
take into account  any
network constraints. 
Nonetheless, the structure of the proof makes it clear that  we exploited only some
properties of the allocation problem. Therefore, the optimal mechanism
construction is valid for any market for which the allocation problem
satisfies these properties. We detail this argument in the next section.

In addition,  the optimal mechanism
construction is valid for  limiting case in which $r=0$ at some edges. In
this case, one needs to specify the definition of  $q$  since the
solution of the allocation problem may not be a singleton.

\section{Extension to General Network Constraints}
We now explain why the optimal mechanism proposal can be extended to a large variety of network constraints. 
This extension is of particular importance for power market networks, since the admissibility of an allocation is subject to its physical feasibility. While this difficulty can be avoided as long as the network is radial, the general case is known to bring its lot of technical challenges. 
As argued in \cite{Philpott2012}
, the allocation problem can be written:

\begin{equation}
\label{eq:counterflow}
    \begin{array}{c}{ \min \sum_{i} J_{i}
    \left(q_{i}\right)} \\ {\text { s.t. } g_{i}(h)+q_{i}=d_{i}, \quad i=1,2, \ldots, n} \\ {A h+B q=b} \\ {h \in H, \quad q \in Q}\end{array}
\end{equation}

With $g_i$ being concave functions,
$A$ and $B$ are $p\times m$ and $p\times n$ real matrices; $b\in\R^p$, and $H$ and $Q$ are (convex) products of segments  in $\R^m$ and $\R^n$ respectively.

They observe that whenever the multipliers of the first set of constraints in 
\eqref{eq:counterflow} are positive (*), then the optimal solution of \eqref{eq:counterflow} is also a solution of 
\begin{equation}
\label{eq:counterflow2}
    \begin{array}{c}{ \min \sum_{i} J_{i}
    \left(q_{i}\right)} 
  \text { s.t. } (h,q)\in C \end{array}
    \end{equation}
    where
    \begin{equation}
    \begin{array}{c}
    C =\{(h,q)\in X
    {\text { s.t. } g_{i}(h)+q_{i}\geq d_{i},
    \quad i=1,2, \ldots, n} \}\end{array}
\end{equation}

and 
    \begin{equation}
    \begin{array}{c}
    X =\{(h,q)
    \text { s.t. }  {A h+B q=b}  {h \in H, \quad q \in Q}\}.\end{array}
\end{equation}
We assume that (*) is satisfied. 
Otherwise said, we require the sub-gradient of the value of \eqref{eq:counterflow2} to be positive. 
We could equivalently require the uniqueness of the solution of \eqref{eq:counterflow2}.

 We can then focus on the study of \eqref{eq:counterflow2}.
We denote by $\delta_C$ the support function of $C$ and set $U=\left\{u=\left(u_{1}, \ldots, u_{n}\right) | u_{i} \leq 0 \right\}$
Applying Theorem 10.1 from \cite{Rockafellar2009},
we get that a necessary and sufficient condition for an  allocation to be optimal is that:
\begin{equation}
\label{eq:FOC1}
    0\in \partial \sum_i  J_i(q_i) + \delta_C(h,q), 
\end{equation}
Now observe that 
\begin{eqnarray}
\partial  \delta_C(h,q) = N_C(h,q) = \\
\left\{z - \sum_i
y_{i} \nabla (g_{i}(h) +q_i)(h,q_i) | y \in N_{U}([g_i(h) +q_i]_i),
z \in N_{X}(h,q)\right\}
\end{eqnarray}
The last equation requires 
the qualification constraint (Q) from \cite{Rockafellar1993} to be satisfied, so one can use  Theorem 4.3 from \cite{Rockafellar1993}.
Still, note that no matter Q being satisfied, $N_C$ does not depends on $c$.
\begin{theorem}
If (*) 
is satisfied, then $q_i^j$ is independent of $c_i^{k}$ for $k\neq j$, 
moreover, the  mechanism proposed ib the previous section can be applied to \eqref{eq:counterflow}. 
\end{theorem}
\begin{proof}
Fix $c$ and consider the associated optimal allocation $\mathbf{q}$.
Observe that either $\mathbf{q}_{i}^{j}\in ]0,\bar{q}_i^j[$ or $\mathbf{q}_i^j\in \{0,\bar{q}_i^j\}$.

\paragraph{First case}
If $\mathbf{q}_i^j\in ]0,\bar{q}_i^j[$, take $k\neq j$
then  $c_i^k$ does not appear in the first order condition \eqref{eq:FOC1}.
By Berge's Maximum Principle \cite{Berge1997}, the optimal allocation is upper hemicontinuous with respect to the parameter $c_i^k$, by unicity of the solution of \eqref{eq:counterflow2}, we get that $q_i$ is continuous with respect to $c_i^k$. 
Thus there is a neighbourhood of $c_i^k$ such that $q_i^j$ is still in $ ]0,\bar{q}_i^j[$. 
In this neighbourhood, condition  \eqref{eq:FOC1}
is satisfied for $q_i^j = \mathbf{q}_i^j$, by unicity of the solution, $q_i^j $ is constant with respect to $c_i^k$ on this neighbourhood.

\paragraph{Second case}
$\mathbf{q}_i^j\in \{0,\bar{q}_i^j\}$. 
Without loss of generality, let us assume that 
$\mathbf{q}_i^j=\bar{q}_i^j$.
Here, we need to observe that  the sub-differential of the criteria with respect to $q_i$ is $[c_i^j,c_i^{j+1}]$, thus the reasoning of the first case can be reproduce whenever $k\neq j+1$. 
So we only need to deal with the situation where $k =j+1$.
Moreover, since $q_i$ is non-increasing in $c_i^{j+1}$, only an increase of $c_i^{j+1}$ can potentially trigger a change in $q_i^j$. Observe that  by Berge's Maximum Principle, $q_i^j$ is continuous with respect to the parameter of interest $c_i^{j+1}$.
If it happens to take a value different than $\bar{q}_i^j$, then this value is also  a solution to \eqref{eq:FOC1} for the initial parameters, which is in contradiction with the unicity of the solution of \eqref{eq:FOC1}. 
\end{proof}

\subsection{Discussion on the non-overlaping zone assumption}

We show here  how the non-overlaping zone structure naturally emerges if we  envision slightly different context  and adapt the notations accordingly:

(1) we focus on a specific producer, and refer to him implicitly in this paragraph, 
(2) we assume the competition is known,
(3) we only suppose the types distribution to have a density $f$ over $C$ \footnote{we still keep non decreasing  marginal costs, otherwise  we would loose the continuity of $q$  }, denoting by $f_j$ the  marginals,
(4) it will prove to be convenient to  use  $q*c:= C^c(q)$ and denote by $j(q)$  the integer part of  $q /\bar{q}$.

Let  $q$ and $x$ be some optimal  allocation and payment  rules for the producer. We assume $q$ and $x$ are continuous and have derivatives almost everywhere.
The producer's profit is $x(s)-  q(s)*c$ whenever his type is $c$   and he signals himself as  of  type $s$.
\textbf{When the competition is known}, we  pinpoint that for any value $q=q(c)$ of the  allocation function, there should be a unique payment that we denote by $x_q$. If it was not the case,  then the  producer would be  better off  revealing what is required to get the  highest possible payment at $q$, which violates  the incentive compatibility constraints (IC).

We observe that, by (IC),  $x(q) - q*c$ is maximal at $q_c$, which implies in particular that   $c_{j(q)}\in \partial_q x(q_c)$. Hence,  almost everywhere and whenever $q$ is not a multiple of $\bar{q}$,  $c_{j(q)}$ is uniquely defined.   Setting $c_q = c_{j(q)}$, \textbf{we observe that 
$x_q = \int_0^q c_t \dd t $}($\star$). 

If we take $s\in C$ such that $q(s)$ is not  a multiple of $\bar{q}$, then by (IC) we know that  $c \rightarrow x(c) -q(c) *s$ should be maximal at $s$. Take $t$ such that $t_{j(q(c))}=s_{j(q(c))}$ and  $q(t)\neq q(s)$, then we should have  $j(q(t))\neq j(q(s))$ which would imply a discontinuity of $q$ somewhere. 
Hence: \textbf{whenever $s_{j(q(s))} = t_{j(q(s))}$, $q(s) = q(t)$} ($\star\star$).
Let $C_k = \{c \in C, j(q(c)) = k \}$, 
$c_k^- = \inf \{c_k, c \in C_k\}$, and 
$c_k^+ = \sup \{c_k, c \in C_k\}$.
By combining ($\star\star$) with the monotony of $q$, we get that 
\textbf{the intersection of the $[c_j^-, c_j^+]$  has an empty interior.}

\subsection{Examples with log-concave functions}
We point out that a sufficient condition to check the monotone
likelihood ratio property is that $F / f$ is increasing. 
If $F$ is a smooth cumulative distribution function with $f$ being the
corresponding smooth and positive density, then
$F/f $ is increasing iff  $f/F$ is decreasing iff $\ln
F'$ is decreasing iff $\ln F$ is concave.
A function $f$ is said to be \emph{log-concave} if $\ln f$ is concave.
Many density functions encountered in  economic and engineering
literature  are \emph{log-concave}: the uniform, the normal, the
exponential and the power function and the Laplace distribution all have
log-concave density functions. 
We refer to \cite{bagnoli2005log} for the results we use on this class
of functions. 
The class of \emph{log-concave} is stable by monotonic
transformation and truncation. Moreover, it happens that if a probability density
distribution is log-concave, then the corresponding cumulative
distribution is log-concave.
In mechanism design theory, it is standard to assume $F$ is 
log-concave \cite{laffont1988dynamics}.

We want to see the implication of the \emph{discernability
  assumption}. 
This assumption imposes a gap $\Delta$ equals to  $K_i^j(c_i^{j+})$ between
$c_i^{j+}$ and $c_i^{(j+1)-}$. We compute this gap for some
standard cases.
To simplify the notations and the computation, we assume without loss
of generality that $c^{j-}=0$ and write $c^{j+}=c^+$. This results in  the following table:
\begin{table}[H]
\centering
\caption{The gap $\Delta$ for some standard probabilities}
\label{my-labelX}
\resizebox{13cm}{!}{
\begin{tabular}{|l|l|l|l|l|}
\hline
Name                   & $\propto f(x)$ & $\propto F(x)$ & $K(x)$ & $\Delta$ \\ \hline
Uniform                & $1$         &   $x$      &   $x $    &  $c^+$        \\ \hline
Power Function        & $\lambda (\frac{x}{c^+})^{\lambda-1}$         &$ c^+(\frac{x}{c^+})^{\lambda} $       &   $\frac{x}{\lambda} $    &     $\frac{c^+}{\lambda} $      \\ \hline
Weibull  & $\lambda (\frac{x}{c^+})^{\lambda-1}e^{(-\frac{x}{c^+})^\lambda} $        &  $ c^+(1-e^{-(\frac{x}{c^+})^\lambda})$      & $\frac{c^+}{\lambda}   (\frac{c^+}{x})^{\lambda-1}(e^{(\frac{x}{c^+})^\lambda}-1)$      &    $c^+ \frac{e-1}{\lambda} $    \\ \hline
Laplace      &$\frac{1}{2}e^{-\lambda |x-\frac{c^+}{2}|} $        &
                                                                
                                                         $x>\frac{c^+}{2}$,
                                                           $\frac{2-
                                                                    e^{-\lambda\frac{c^+}{2}}
                                                                        e^{-\lambda(x-\frac{c^+}{2})}}{2\lambda}$&
                                                                                    
                                                                  &      $\frac{2}{\lambda}(e^{\frac{c^+}{2}\lambda}-1)
                                                                                                           $     \\ \hline
Exponential (reversed) & $\lambda e^{-(c^+-x)\lambda}$        & $e^{-c^+\lambda}(e^{x\lambda}-1) $       &   $\frac{1-e^{-x\lambda}}{\lambda}$     &    $\frac{1-e^{-c^+\lambda}}{\lambda}$        \\ \hline
\end{tabular}
}
\end{table}

We truncate the probabilities so that they have support in $[0,c^+]$. 
The symbole $\propto$ means that we express $f$ and $F$ modulo the
multiplication by a common constant (due to the truncation) and
$\lambda$ is a positive parameter that should be greater than $1$ for
the Power function and the Weibull probability.
For the uniform distribution, we see that the intervals should be of
non-decreasing sizes. For instance, one could take $c^1\in
[\bar{c},2\bar{c}]$,   $c^2\in
[3\bar{c},4\bar{c}]$, $c^3\in [5\bar{c},6\bar{c}]$, etc.
For the Power,   the Weibull  and the exponential functions, 
we see that the gap could be made smaller.

\section{Study of the allocation problem}
\label{sec:nagents}
\subsection{The standard auction problem}
The previous section motivates the study of the allocation problem for different reasons. Firstly, as we have already pointed out in the proofs, the results of \S \ref{sec:mechanism} rely on some properties of the solution of the standard allocation problem. In addition to those properties, we derive in this section an algorithm to compute the solution of the standard allocation problem.   
According to \ref{th:mainResult}, these algorithms can be used for both the original auction problem and the optimal mechanism design. 
To benchmark the  mechanism design equilibrium against an equilibrium of the Bayesian game related to the standard auction, numerical efficiency is pivotal: indeed, the Bayesian equilibrium estimation requires a large number of  allocations computations.

Let us first introduce the standard allocation problem. In a standard mechanism, the principal solves an allocation problem
based on the bids he receives. 
Those bids will be denoted  by $c_i^j$, where as before $i\in I$ corresponds to
the  ith agent and $j\in J$ corresponds to the jth  working zone with
constant marginal price. To model the fact that the production costs
are piecewise linear, we use some positive variables $q_i^j$ so that
$q_i^j\leq \bar{q}$, for any $i\in I$, the quantity produced by agent
$i$ is  $q_i=\sum_{j\in J} q_i^j$ and the related production cost is
$\sum_{j\in J} c_i^j q_i^j$.
As before, an allocation should satisfy the constraint that production  exceeds demand. 
We end up with Problem \ref{pb:alloc}:
\label{sec:fixed-point-theorem}
\begin{problem}
\label{pb:alloc}
\begin{equation} 
\begin{aligned}
& \underset{(\q,h)}{\text{minimize}}
& & \sum_{i\in I}\sum_{j\in J} q_i^j c_{i}^j  \\
& \text{subject to}
& &\forall i\in I:  \sum_{j\in J } q_i^j  + \sum_{i'\in V(i)}
h_{i',i} - h_{i,i'} - \frac{h^2_{i,i'}+h^2_{i',i}}{2}r_{i,i'} \ge d_i  \quad (\lambda_i) \\
& & &\forall (i,i')\in E: h_{i,i'} \ge 0  \quad (\gamma_{i,i'})\\
& & &\forall (i,j) \in I\times J: q_i^j\ge 0  \quad (\mu_{i,j})\\
&&&\forall (i,j)\in I\times J:  q_i^j \leq \bar{q} \quad (\nu_{i,j}).
\end{aligned}
\end{equation}
\end{problem}
The notations for the dual the variables associated with each
constraint are indicated in parentheses. Those variables are in $\R_+$.

For any node $i\in I$, we define the function $F_i$ for $\lambda \in [\min_{i} c_i^1,\max_i c^N_i]^n$
\begin{equation}
  F_i(\lambda_i, \lambda_{-i}) =d_i + \sum_{i'\in V(i)} \frac{\lambda_{i'}
  - \lambda_i}{r_{i,i'} (\lambda_i +\lambda_{i'})} +\frac{(\lambda_{i'}
  - \lambda_i)^2}{2r_{i,i'} (\lambda_i +\lambda_{i'})^2}.
\end{equation}
Later on we justify that this function could be interpreted as the production of agent $i$ when the
multipliers are $\lambda_i$ and $\lambda_{-i}$.
Its partial derivative with respect to $\lambda_i$ is 
\begin{equation}
\label{eq:der1}
 \partial_{\lambda_i}F_i(\lambda_i,\lambda_{-i}) = - \sum_{i'\in V(i)}
  \frac{4}{r_{i,i'}} \frac{\lambda_{i'}^2}{(\lambda_i+\lambda_{i'})^3}<0.
\end{equation}
The derivative is negative: when $i$ increases its price
 it is assigned smaller production quantities.
The partial derivative of $F_i$  for $i'\in I\backslash \{i\}$ is 
\begin{equation}
\label{eq:der2}
 \partial_{\lambda_{i'}}F_i(\lambda_i,\lambda_{-i})  =\begin{cases}
 \frac{4}{r_{i,i'}} \frac{\lambda_{i'}
    \lambda_i}{(\lambda_i+\lambda_{i'})^3}>0 & \text{if }   i'\in V(i)\\
0 & \text{else. } 
\end{cases}
\end{equation}
When another agent becomes less competitive, $i$ is assigned more
production. 
Let $k\in J\cup \{0\}$.
The limit at $+\infty$ and $0$ of  $F_i(x,\lambda_{-i})-k\bar{q} $ are
\begin{equation}
\lim_{x\rightarrow +\infty}F_i(x,\lambda_{-i})-k\bar{q} = d_i -k\bar{q} -\sum_{j\in V(i)}\frac{1}{2 r_{i,j}} 
\end{equation}
and    
\begin{equation}
\lim_{x\rightarrow +\infty}F_i(x,\lambda_{-i})-k\bar{q} =d_i -k\bar{q}+\sum_{j\in V(i)}\frac{3}{2 r_{i,j}}.
\end{equation}
Without loss of generality (otherwise we could impose capacity constraints), we assume, the first term to be strictly negative and the second to be strictly positive, hence  
by the intermediate value theorem, $F_i -k\bar{q}$ has a zero. 
Since $F_i -k\bar{q}$ is decreasing in $\lambda_i$, this solution is unique.
Now we define for $i\in I$ and $k\in J\cup \{0\}$, $g^k_i$ as the function that
associates any $\lambda_{-i}  \in [\min_{i} c_i^1,\max_i c^N_i]^{n-1}$ 
with the unique $x$ such that
and 
$F_i(x,\lambda_{-i}) = k\bar{q} $ and $ x> 0$:
\begin{equation}
\begin{aligned}
&F_i(g^k_i(\lambda_{-i}),\lambda_{-i}) = k\bar{q} \\
&g^k_i(\lambda_{-i})> 0.
 \end{aligned}
\end{equation}
\begin{lemma}
For any $i\in I$, $k\in J\cup \{0\}$, $\lambda^{-i}\in [\min_{i}
c_i^1,\max_i c^N_i]^{n-1}$ and $i'\in V(i)$
\begin{equation}  
  \partial_{\lambda_{i'}} g_i^k (\lambda_{-i}) > 0. 
\end{equation}
In particular, $g_i^k$ is increasing in $\lambda_{i'}$ for $i'\in V(i)$.
\end{lemma}
\begin{proof}
According to the implicit function theorem
\begin{equation}
\label{eq:implicite}
  \frac{\partial g_i^k (\lambda_{-i})}{\partial \lambda_{i'}}  =  - \frac{\partial F_i }{\partial \lambda_{i'}} /  \frac{\partial F_i}{\partial \lambda_i} ,
\end{equation}
\end{proof}

It is clear that $g_i^k(\lambda_{-i})$ is decreasing in $k$.
We  proceed with the  computation of the dual of Problem
\ref{pb:alloc}. 
If a strong duality theorem applies, then we should have
\begin{eqnarray*}
 \min_{q,h}\max_{\lambda,\gamma,\nu ,\mu}\sum_{i\in I,j\in J} q_i^j c_i^j  +\\
\sum_{i\in I} \lambda_i \{d_i - (\sum_{j\in J} q_i^j  + \sum_{i'\in V(i)}
h_{i',i} - h_{i,i'} - \frac{h^2_{i,i'}+h^2_{i',i}}{2}r_{i,i'} )\}\\
-\sum_{i \in I,j\in J}  \gamma_{i,j} h_{i,j}   +
\sum_{i\in I,j \in J} \nu_{i,j}(q_i^j -\bar{q})  - \mu_{i,j}q_i^j \\
 =\max_{\lambda,\gamma,\nu\mu} 
 \min_{q,h}
 \sum_{i\in I} \lambda_i d_i-\sum_{i\in I,j\in J} \nu_{i,j}\bar{q}+
q_i^j(c_i^j +\nu_{i,j}  - \lambda_i- \mu_{i,j})\\
 + \sum_{(i,i')\in E}
h_{i,i'}\{\lambda_i-\lambda_{i'}- \gamma_{i,j}\} +
h^2_{i,i'}r_{i,i'} \frac{ \lambda_i+\lambda_{i'} }{2},\\
\end{eqnarray*}
so that for any $(i,i')\in E$, by necessary and sufficient first order condition 
\begin{eqnarray}
\label{eq:h}
h_{i,i'} = \frac{\gamma_{i,i'}+ \lambda_{i'}-\lambda_{i} }{r_{i,i'}(\lambda_{i'}+\lambda_{i} )}.
\end{eqnarray}

By replacing $h$ by its expression in the dual variables we get something equivalent to 
\begin{equation} 
\begin{aligned}
& \underset{(\lambda,\gamma,\mu,\nu)}{\text{maximize}}
& & \sum_{i\in I} \{\lambda_i d_i-\sum_{j\in J} \nu_{i,k}\bar{q} -
\sum_{i'\in V(i)}
\frac{(\lambda_i-\lambda_{i'}- \gamma_{i,j})^2}{2r_{i,i'}(\lambda_i+\lambda_{i'} )} \} \\
& \text{subject to}
& & \forall (i,j)\in I\times J \quad c_i^j +\nu_{i,j}  \geq \lambda_i+ \mu_{i,j}.
\end{aligned}
\end{equation}
The expression of $\gamma$ with respect to $\lambda$ follows. For any
$(i,i')\in E$
\begin{equation}
 \gamma_{i,i'}=\begin{cases}
0 & \text{if }  \lambda_i\le\lambda_{i'}\\
\lambda_i-\lambda_{i'}& \text{else} 
\end{cases}
\end{equation}
thus the dual problem is equivalent to  
\begin{equation} 
\begin{aligned}
& \underset{(\lambda,\mu,\nu)}{\text{maximize}}
& & \sum_{i\in I} \{\lambda_i d_i-\sum_{j\in J} \nu_{i,j}\bar{q} -
\sum_{i'\in V(i)}
\frac{(\lambda_i-\lambda_{i'})^2}{4r_{i,i'}(\lambda_i+\lambda_{i'} )} \} \\
& \text{subject to}
& &\forall (i,j)\in I\times J \quad c_i^j +\nu_{i,j}  \geq \lambda_i+ \mu_{i,j},
\end{aligned}
\end{equation}
because $\mu$ does not play any role in the admissibility of the other
variables nor in the objective, this is equivalent to 
\begin{equation} 
\begin{aligned}
& \underset{(\lambda,\nu)}{\text{maximize}}
& & \sum_{i\in I} \{\lambda_i d_i-\sum_{j\in J} \nu_{i,j}\bar{q} -
\sum_{i'\in V(i)}
\frac{(\lambda_i-\lambda_{i'})^2}{4r_{i,i'}(\lambda_i+\lambda_{i'} )} \} \\
& \text{subject to}
& & \forall (i,j)\in I\times J  \quad c_i^j +\nu_{i,j}  \geq \lambda_i,
\end{aligned}
\end{equation}
The expression of $\nu$ follows. For any
$(i,j)\in I\times J$
\begin{equation}
\nu_{i,j}=\begin{cases}
0 & \text{if }  \lambda_i \leq c_{i}^j\\
\lambda_i- c_{i}^j& \text{else.} 
\end{cases}
\end{equation}
We can now  justify that we have strong duality: 
the operator is continuous, convex-concave and the dual variables are
restricted to be in a bounded set.

The dual of the allocation problem is therefore written: 
\begin{equation} 
\begin{aligned}
& \underset{\lambda\geq 0}{\text{maximize}}
& & \sum_{i\in I} \{\lambda_i d_i-\bar{q} \sum_{j\in J} (\lambda_i -
c_i^j)\delta_{\lambda_i \ge c_{i}^j}- \sum_{i'\in V(i)}
\frac{(\lambda_i-\lambda_{i'})^2}{4r_{i,i'}(\lambda_i+\lambda_{i'} )} \}, \\
\end{aligned}
\end{equation}
where 
\begin{equation}
\delta_{x\geq y}=\begin{cases}
1 & \text{if }  x\geq y  \\
0& \text{else. } 
\end{cases}
\end{equation}
For $i \in I$ we maximize the criteria
 \begin{equation}
\label{eq:criteria}
 \lambda_i d_i-\bar{q} \sum_{j \in J} (\lambda_i -
 c_i^j)\delta_{\lambda_i \ge c_i^j}- \sum_{i'\in V(i)}
\frac{(\lambda_i-\lambda_{i'})^2}{4r_{i,i'}(\lambda_i+\lambda_{i'} )},
\end{equation}
which  is strictly concave for any $\lambda_{-i}$ (sum of concave and
strictly concave functions). We denote by $\Lambda_i(\lambda_{-i})$ its  maximizer. 
The first order necessary and sufficient condition on  $\Lambda_i$ is: 
 \begin{equation}
\label{eq:FOC}
0\in F_i(\Lambda_i, \lambda_{-i}) - K_i(\Lambda_i),
\end{equation}
where 
\begin{equation}
K_i(\lambda_i)=\begin{cases}
0 & \text{if }  \lambda_{i} <  c_i^1 \\
[j-1,j]\bar{q} & \text{if }  \lambda_{i} =  c_i^j \\
j\bar{q} & \text{if } \lambda_{i}\in ]c_i^j,c_i^{j+1}[,
  j\neq N\\
N\bar{q} & \text{if } \lambda_{i}\in \lambda_{i}\in ]c_i^N,\bar{c}[, 
\end{cases}
\end{equation}
We conclude 
\begin{lemma}
\label{lemma:unicity}
For any $i\in I$ and any $\lambda^{-i}\in [\min_{i} c_i^1,\max_i c^N_i]^{n-1}$,   $\Lambda_i(\lambda_{-i}) $ is the unique solution of
\begin{eqnarray} 
F_i(\Lambda_i, \lambda_{-i}) \in K_i(\Lambda_i).
\end{eqnarray}
\end{lemma}
We point out that the primal (and dual) solution unicity is a
desirable property that is not
systematic for the allocation problems of centralized market
models.
The expression of $h$ with respect to $\lambda$ \eqref{eq:h} together with the fact the fact the
supply constraint should be binding at optimality justify the
interpretation of $F_i$ proposed at the beginning of this
subsection. In the following sequel we use this property many times.

\subsection{Some properties of the solution}
\label{sec:property}
If $r $ and $d$ are set, we can see the solution of Problem 4 as a
function of the vector $c\in\CC^n$. We denote by $q(c)$ the solution of Problem 4 with the
cost vector  $c$. Similarly, we define $q_i(c)$, $q_i^j(c)$, $\lambda(c)$ and $\lambda_i(c)$.
We give here two properties of the allocation problem solution. 
By
integration, we showed in the previous section that the solution of
the mechanism design inherits those properties.

\begin{lemma}
\label{lemma:indep1}
Let $(q(\cc),h(\cc))$ be a solution of Problem \ref{pb:alloc}, then $q_i^j(\cc)$
does not depend on $c_i^{l}$ for $l\neq j$:
\begin{equation}
 q_i^j(c^1,\ldots c^{j-1},c^j,c^{j+1}\ldots,c^N;c^{-i}) =  q_i^j(s^1,\ldots s^{j-1},c^j,s^{j+1}\ldots,s^N;c^{-i}) 
\end{equation}
\end{lemma}
\begin{proof}
Let $i\in I$, $j\in J$, $c_{-i}\in\CC^{n-1}$, $c=(c^1,\ldots,c^N)\in \CC$ and 
$s=(s^1,\ldots,s^N)\in \CC$ such that $s^j=c^j$. We shall prove that $q_i^j(s,c^{-i}) = q_i^j(c,c^{-i})$.
We denote by $\lambda^c$ (resp. $\lambda^s$) the dual variables associated with the nodal contraints for the allocation problem parametrized with $c$ (resp. $s$).
First if 
\begin{equation}
q_i^j(c,c^{-i})\in ]0,\bar{q}[,
\end{equation}
then by lemma \ref{lemma:unicity} 
$\lambda^c_i = c_i^j$ and  using Lemma \ref{lemma:unicity} again,
$\lambda^s_i = c_i^j$. Therefore $\lambda^s= \lambda^c$, from which we
deduce  that $ q_i^j(c,c^{-i}) =  q_i^j(s,c^{-i}) $.

Therefore without loss of generality, we can assume that 
\begin{equation}
\label{eq:24}
q_i^j(c,c^{-i})=\bar{q}\quad\mbox{and}\quad
q_i^j(s,c^{-i})=0.
\end{equation}
Then using Lemma \ref{lemma:unicity} 
we get 
\begin{equation}
\lambda^c_i \geq c^k\quad\mbox{and}\quad
\lambda^s_i\leq c^k,
\end{equation}
so that $\lambda^c_i\geq \lambda^s_i$. 
If $\lambda^c_i > \lambda^s_i$, then $\lambda^c_{-i} \geq
\lambda^s_{-i}$ by non-decreasingness of $\Lambda_{i'}$, $i' \in
I\backslash \{i\}$ (explained in \S\ref{sub:fixedpoint})
Therefore all the other agents are producing less, which is absurd since
$i$ is already producing less.  

\end{proof}

We extend the notations by setting for all $i\in I$, $c_i^0=c_*$.
We consider the subset $\Set$ of $\CC$ for which at  some
nodes $i$, the multiplicator $\lambda_i$ is equal to the marginal cost \emph{and} the
production is a multiple of $\bar{q}$ (i.e. stuck in an angle):
\begin{equation}
\Set = \{ c\in\CC^n, q_i(c) = j\bar{q} \ \mbox{and} \
\lambda_i(c) = c_i^{j'} \ \mbox{for some} \  i\in I, j\in J\cup\{0\},
j'\in \{j,j+1 \}\}.
\end{equation}
The set $\Set$ corresponds to the points of transition  between the two possibilities
defined by the first order condition \eqref{eq:FOC}. Because of the
angle, it is natural to think that this is where irregularities may
happen (see the proof of the next lemma).
We introduce this set to show some regularity properties of $q$ and
$Q$. 
We detail the proof in the Appendix.  The approach consists in showing
that $\Set$ is a finit  union of sets of zero measure.
This is also true for the projection of $\Set$ on the $\{c_i\}\times\C^{-i}$. 
Then we observe that on $\CC\backslash \Set$, the relations between  the
primal and dual variables are smooth. 
\begin{lemma}
\label{lemma:qc1}
The function $q$ is $C^\infty$ on $\CC^n  \backslash \Set $ and $C^0$
on $\CC^n$.
\end{lemma}
\begin{proof}
We postpone the proof to Appendix \ref{appendix:qc1}
\end{proof}
\subsection{Fixed point}
\label{sub:fixedpoint}
In this subsection we show that the solution of the dual problem is
the unique fixed point of a monotone operator.
We define 
\begin{equation}
\Lambda(\lambda_1,...,\lambda_n) = (\Lambda_1(\lambda_{-1}),...,\Lambda_n(\lambda_{-n})).
\end{equation}

\begin{lemma}
For any $i\in I$, $\Lambda_i$ is non-decreasing.
\end{lemma}

\begin{proof}
Let  $\lambda_{-i}<\lambda_{-i}'$ and the corresponding $\Lambda_i$ and $\Lambda_i'$.
Assume $\Lambda_i >\Lambda_i'$. 
Since $F_i$ 
is  decreasing in the first variable and increasing in the second  
\begin{equation}
F_i(\Lambda_i, \lambda_{-i})  < F_i(\Lambda_i', \lambda_{-i}') 
\end{equation}
Moreover for any $x \in K(\Lambda_i')$ and $y \in K(\Lambda_i)$, $x\leq y$
and
$F_i(\Lambda_i, \lambda_{-i}) \in K(\Lambda_i)$,
$F_i(\Lambda_i', \lambda_{-i}') \in K(\Lambda_i')$.
Therefore 
$  F_i(\Lambda_i', \lambda_{-i}') \leq F_i(\Lambda_i, \lambda_{-i})$
which is absurd. 
\end{proof}

We will use the following classical result (see
\cite{topkis1998supermodularity} for a proof and  definition of 
complete lattice). 
\begin{theorem}[Knaster-Tarski fixed point]
Let $L$ be a complete lattice and let $f$ an application from $L$ to $L$ and order preserving. 
Then the set of fixed points of $f$ in $L$ is a complete lattice.
\end{theorem}

In particular, the set of fixed points is non empty.
Since $\Lambda$ is order preserving and $[c_*,c^*]^n$ is a lattice when we consider the natural order, there is a fixed point, and the set of fixed points is a lattice.
\begin{lemma}
$\lambda$ is optimal for the dual $\Leftrightarrow$ $\lambda$ is a fixed point of $\Lambda$.
\end{lemma}
\begin{proof}
\begin{itemize}
\item If $\lambda$ is optimal for the dual, then each component $i$ maximizes
the criteria \eqref{eq:criteria}, thus $\lambda$   is a fixed point of $\Lambda$.
\item If $\lambda$ is a fixed point of $\Lambda$, then by definition,
  each component $i$ maximizes the criteria \eqref{eq:criteria}. Hence since the problem is (strictly) concave,
$\lambda$ is optimal.
\end{itemize}
\end{proof}

A consequence of the previous lemma is that
\begin{lemma}
The set of fixed points of $\Lambda$ is a singleton.
\end{lemma}
\begin{definition}[Continuous for monotone sequence]
We consider the natural partial order on $\R^n$.
We say that a function $G$ is continuous for monotone
(resp. increasing, decreasing) sequences if for
any monotone (resp. increasing, decreasing) sequence $x_n$ converging to a point $x$ in the domain of $G$,
$G(x_n)$ goes to $G(x)$ as $n$ goes to infinity. 
\end{definition}

Clearly, a function is continuous for monotone sequences if and only
if it is continuous for  increasing and decreasing sequences.

\begin{lemma}
The operator $\Lambda$ is continuous for monotone sequences.
\end{lemma}

The intuition of the proof is that we can use the monotony of the
sequence and Lemma \ref{lemma:unicity} to characterize the behaviour of
$\Lambda$ on the neighborhood.  We find that $\Lambda$ is either 
constant or characterized by 
the implicit function theorem.
\begin{proof}
\newcommand{\lambi}{\bar{\lambda}_{-i}}
\newcommand{\La}{\Lambda}
Let $\bar{\lambda}_{-i}$, $j\in[1\ldots N]$, we first deal with the 'nice' case, that
corresponds to $F_i(\La (\lambi ),\lambi) \in
]j-1,j[\bar{q}$

\begin{itemize}
\item If $\La_i (\lambi)\in ]c_i^j,c_i^{j+1}[$ (we do not treat the case
$j=N$, which is very similar to what follows)
then since $F_i$ is $C^{\infty}$ and of invertible   derivative (non
zero) in $\lambda_i$,
the implicit function theorem 
 tells us that the solution $\psi$ of $F_i(\psi(\lambi),\lambi)=
 j\bar{q}$ is continous in a neighborhood $B$ of $\lambi$.
Thus we can make $B$ small enough so that for $\lambda_{-i}\in B$,
$\psi (\lambda_{-i}) \in ]c_{i}^j,c_{i}^{j+1}[$. On this neighborhood,
$\psi$ satisfies the first order conditions and so  by unicity of the
solution of the optimization problem, since those conditions are
sufficient,  $\psi = \La_i$ on $B$.
Therefore $\La_i$ is continous  at $\lambi$.

\item If $\La_i (\lambi) = c_{i}^j$
(as before, we do not treat the case
$j=N$),
then by Lemma \ref{lemma:unicity} $F_i(\La_i (\lambi),\lambi) = [j-1,j]\bar{q}$, if $F_i\in
]j-1,j[\bar{q}$ (we  deal with the border case in the next point)
then since $F_i$ is continuous, there is a neighborhood $B$ of $\lambi$ such that  
$F_i(\La_i (\lambi),\lambda_{-i}) \in ]j-1,j[\bar{q}$, so on $B$ $\La_i$ is
constant and therefore continuous.

\item We proceed with the borders.
If $F_i(\La_i (\lambi), \lambi) = (j-1) \bar{q}$ and $\La_i (\lambi)
  = c_{i}^j$.

\begin{itemize}
\item Decreasing case:
Let us take $\epsilon\in \R^{n-1}_+$ such that $F_i(\La_i (\lambi), \lambi
+\epsilon) \in [j-1,j]\bar{q}$ ($F_i$ is continuous and increasing in $\lambda_{-i}$).
Then $\La_i(\lambi +\epsilon) = \La_i (\lambi)$ checks the first order
condition so $\La$ is constant, and we get the continuity for decreasing
sequences.

\item Increasing case: 
$F_i(\La_i (\lambi), \lambi) = (j-1) \bar{q}$  hence there exists a
ball $B$ such that  the implicit function theorem applies and
there exists $\psi$ such that $F_i(\psi(\lambi-\epsilon),
\lambi-\epsilon) = (j-1) \bar{q}$ and $\psi(\lambi)= \La_i (\lambi)=c_{i}^{j}$
(remember that $\La_i (\lambi) = c_{i}^{j}$ by hypothesis) .
Since $F_i$ is 
increasing in the second variable and decreasing in the first, $\psi$
is increasing. 
For $\epsilon$ of positive components and
sufficiently small, $\psi(\lambi-\epsilon) \in ]c_{i}^{j-1},c_{i}^{j}[$
(since $\psi(\lambi)= \La_i (\lambi)=c_{i}^j$) and  check the first
order condition.
Therefore  for $\epsilon$ of positive components and
sufficiently small, $\psi = \La_i$ by uniqueness of the solution.
Thus $\La_i$ is continuous for increasing sequence.
\end{itemize}

\item We do the same analysis    if $F_i(\La_i (\lambi), \lambi) = j \bar{q}$ and $\La_i (\lambi)
  = c_{i}^j$.
\end{itemize}
The conclusion follows.
\end{proof}

We could have alternatively used the Berge Maximum theorem for strictly
concave criterion to get the continuity of $\Lambda$.
Yet, we chose to present this proof for pedagogical reasons because it
contains some key ideas (see appendix). 

\begin{theorem}
The sequence $(\Lambda^k(c_1^N...c_n^N))_{k\in\N} $ converges to the solution of the dual.
\end{theorem}
\begin{proof}
Since $\Lambda(c_1^N...c_n^N) \leq (c_1^N...c_n^N)$, and since
$\Lambda$ is order preserving,  the sequence $\Lambda^k(c_1^N...c_n^N) =\lambda^k$
is non increasing and bounded, therefore it converge to a point $x$. Since $\Lambda$
is continuous for monotone sequence, 
$x$ is a fixed point.
\end{proof}

\begin{theorem}
For any $i\in I$, $\lambda_{-i}\in [c_*,c^*]^{n-1}$, 
$\Lambda_i(\lambda_{-i})$ has the  following  explicite expression:
\begin{equation}
\label{eq:G}
\Lambda_i(\lambda_{-i})=\min \{ c_i^N, 
\min_{j\in J} \{c_i^j 1_{F_i( c_i^j,\lambda_{-i})< j\bar{q}}\},
\min_{k\in[0..N-1]}\{ g_i^k(\lambda_{-i}) 1_{g_i^k(\lambda_{-i})\in[c_i^k,c_i^{k+1}]}\}\}
\end{equation}
\end{theorem}
\begin{proof}
We denote by  $G_i$ the RHS of \eqref{eq:G} and show that 
for any $i$
\begin{eqnarray}
F_i(G_i(\lambda_{-i}), \lambda_{-i}) \in K(G(\lambda_{-i})),
\end{eqnarray} 
and then we conclude with a uniqueness argument.

If there is $j\in J$ such that  $G_i(\lambda_{-i}) = c_i^j$, then 
either $F_i(c_i^j,\lambda_{-i})<j\bar{q}$ or $g_i^j(\lambda_{-i}) =
c_i^{j}$. This last possibility implies by definition of $g_i^j$ that $F_i(c_i^j,\lambda_{-i}) = j\bar{q}$.
Thus   $F_i(c_i^j,\lambda_{-i}) \leq j \bar{q}$. 
Remember that $K(G(\lambda_{-i})) = [j-1,j]\bar{q}$.
So we need to prove that $F_i(c_i^j,\lambda_{-i}) \geq (j-1) \bar{q}$.
Suppose the contrary, i.e. $F_i(c_i^j,\lambda_{-i}) < (j-1) \bar{q}$.
Then since $G_i(\lambda_{-i}) = c_i^j$, $F(c_i^j,\lambda_{-i})
<(j-1)\bar{q}$, which in turn implies that 
\begin{equation}
\label{eq:g1}
g_i^j(\lambda_{-i}) < c_i^j.
\end{equation}
Now observe that  since  $G_i(\lambda_{-i}) = c_i^j$,
$F(c_i^{j-1},\lambda_{-i}) > (j-1) \bar{q}$, which implies that 
\begin{equation}
\label{eq:g2}
g_i^j(\lambda_{-i}) >c_i^{j-1}.
\end{equation}
Combining \eqref{eq:g1} and \eqref{eq:g2} with the definition of $G$,
we see that $G(\lambda_{-i}) \leq g_i^j(\lambda_{-i})$. But
$G(\lambda_{-i})=c_i^j$ and $g_i^j(\lambda_{-i}) < c_i^j$, which is
absurd.
Therefore $F_i(c_i^j,\lambda_{-i}) \geq (j-1) \bar{q}$.

Otherwise,  let us assume that there is not such $j$. 
Then there is $k\in [0\ldots N-1]$ such that $G_i(\lambda_{-i}) =
g_i^k(\lambda_{-i})$. By definition of $g_i^k$, $F_i(G_i(\lambda_{-i})
,\lambda_{-i}) )= k\bar{q}$ and by definition of $G$,
$G_i(\lambda_{-i})\in [c_i^k,c_i^{k+1}]$.
So  again $ F_i(G_i(\lambda_{-i}), \lambda_{-i})) \in K(G_i(\lambda_{-i}))$.
We can now conclude that $\Lambda = G$.
\end{proof}

We can interpret the fixed point algorithm as if some benevolent agents
situated at each node of the network 
were exchanging information.
They collectively try to minimize the total cost and, to do so, they
communicate their current marginal costs. This marginal cost is the
minimum of their local marginal cost and the marginal cost of
importation from the  adjacent nodes.
At each iteration, the agents compute how much they are going to
produce based on their current marginal cost.
They then update their marginal cost based on the
information they just received and transmit this marginal cost to the
adjacent  nodes.
We point out that the information used by each agent is local.

\subsection{Decreasing Rate}
We derive in this section  an estimate for the decreasing rate. 
We denote 
$\alpha = \max_{(e,e')\in E^2} r_{e}/r_{e'}$.
We have the following bound:
\begin{lemma}
For any $(i,i', k,\lambda_{-i})\in E \times [0,N] \times [c_*,c^*]^{n-1}$, 
\begin{equation}
\partial_{\lambda_{i}}g_{i'}^k(\lambda_{-i})  \geq \frac{1}{N\alpha}(\frac{c_*}{c^*})^5. 
\end{equation}
\end{lemma}
\begin{proof}
We combine \eqref{eq:implicite} with \eqref{eq:der1} and \eqref{eq:der2}.
\end{proof}
\begin{lemma}
Since $(\lambda^k_i)_{k\in\N}$ is non-increasing for all $i\in I$, there is a finite number of $k$ for which
at least 
one coordinate $\lambda_i^k$ satisfies
\begin{equation}
 \lambda_i^k >c_i^q \quad \mbox{and} \quad \lambda_i^{k+1}
\leq c_i^q
\end{equation}
 or  
\begin{equation}
\lambda_i^k = c_i^q \quad \mbox{and} \quad \lambda_i^{k+1}
< c_i^q.
\end{equation}
  We denote by $\K$ this set.
Let  $(k_1, k_2)\in \N^2$ such that $[k_1-1, k_2+1] \cap \K =\emptyset$.
Then for $k\in [k_1,k_2]$ and $i\in I$ such that $\lambda^{k-1}_i \neq
 \lambda^{k}_i $
\begin{eqnarray}
\lambda^{k}_i -\lambda^{k+1}_i 
\geq \frac{1}{N\alpha}(\frac{c_*}{c^*})^5 \quad \max_{i'\in V(i)} ( \lambda^{k-1}_{i'} - \lambda^{k}_{i'} )
\end{eqnarray}
\end{lemma}
\begin{proof}
By definition of $\lambda^k$,
$
\lambda^{k}_{i} -\lambda^{k+1}_{i}  = \Lambda^i(\lambda^{k-1}_{-i} ) -
  \Lambda^i(\lambda^{k}_{-i}).
$
By construction, there exists  $j\in [0,N-1]$ such that
$\Lambda^i(\lambda^{k-1}_{-i} )=g_i^j(\lambda_{-i}^{k-1})$ and
$\Lambda^i(\lambda^{k}_{-i} )=g_i^j(\lambda_{-i}^{k})$.
Then by monotony of $g$,   $g_i^j(\lambda_{-i}^{k})-g_i^j(\lambda_{-i}^{k-1})$ is lower bounded
 by 
\begin{equation}
|\partial_{\lambda_{i'}}g_i^j|_{\infty}  ( \lambda^{k-1}_{i'} - \lambda^{k}_{i'} ),
\end{equation}
for $i'\in V(i)$. 
We then  take the $i'\in V(i)$ that maximizes  $( \lambda^{k-1}_{i'} -
\lambda^{k}_{i'} )$ and use the previous lemma to get the result. 
\end{proof}

\subsection{Algorithm Implementation}
We implemented this algorithm in Matlab. We used a dichotomy to
compute the $g_i^k$.
Note that for linear cost  the analysis is similar.  
We define
$g_i(\lambda_{-i})$ as the unique $x$ such that 
$f_i(x,\lambda_{-i}) = 0 $ and $ x\geq 0 $
and  define $\Lambda$ such that   
\begin{equation}
\Lambda_i(\lambda) = \min(c_i, g_i(\lambda_{-i}))
\end{equation}

We performed some numerical comparisons with CVX, a package for
specifying and solving convex programs \cite{gb08,cvx} for both linear
and piecewise linear production cost functions.
We generated a graph with 100 nodes connected randomly. To generate the
graph, we used a Barabasi-Albert model \cite{barabasi1999emergence} to ensure some scaling
properties. 
The experiment was performed on a personal laptop (OSX, 4 Go,1.3 GHz
Intel Core i5). 
The networks randomly generated to test the implementations are
displayed in Figures \ref{fig:fig1} and \ref{fig:fig2}, and the results are summarized in
Table \ref{table:resultssimu}.

Both CVX and the fixed point algorithm converges to an estimate of  the optimal value. 
We did not try to optimize the numerical algorithm, but some trick could be used to avoid the costly estimation of the $g$s.
Still, the linear version of the fixed point algorithm was about ten times
faster than the CVX resolution.
Note that the algorithm could be distributed, since at each iteration,
the computation at each node only depends on the values of the
previous iteration.

\begin{table}
\centering
\begin{subtable}{.5\textwidth}
\label{table:linear}
\centering
\begin{tabular}{|l|l|l|}
\hline
           & \textbf{Fixed Point} & \textbf{CVX} \\ \hline
\textbf{cost}  &83.2                & 83.195              \\ \hline
\textbf{time (s)} &2.03              & 30.23      \\ \hline
\end{tabular}
\end{subtable}%
\begin{subtable}{.5\textwidth}
\centering
\label{generic}
\begin{tabular}{|l|l|l|}
\hline
           & \textbf{Fixed Point} & \textbf{CVX} \\ \hline
\textbf{cost}  &4971.4             & 4971.4              \\ \hline
\textbf{time (s)} &28.39              & 35.23      \\ \hline
\end{tabular}
\end{subtable}
\caption{Results for a linear (a) and piecewise linear (b) instances
  of the problem solved with the  fixed point algorithm and CVX. }
\label{table:resultssimu}
\end{table}

\begin{figure}
\begin{subfigure}[b]{0.45\textwidth}
  \centering
  \includegraphics[width=\textwidth]{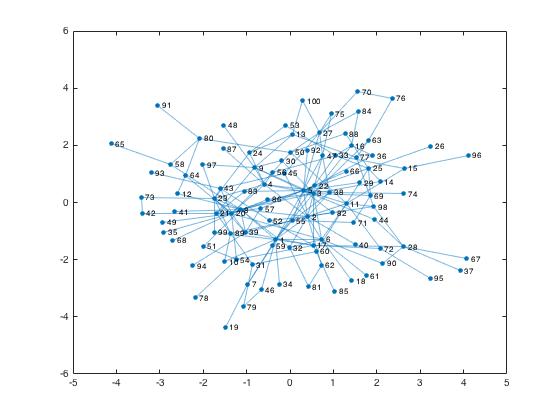}
\caption{The network generated to test the linear implementation of
  the algorithm}
\label{fig:fig1}
\end{subfigure}
\begin{subfigure}[b]{0.45\textwidth}
  \centering
  \includegraphics[width=\textwidth]{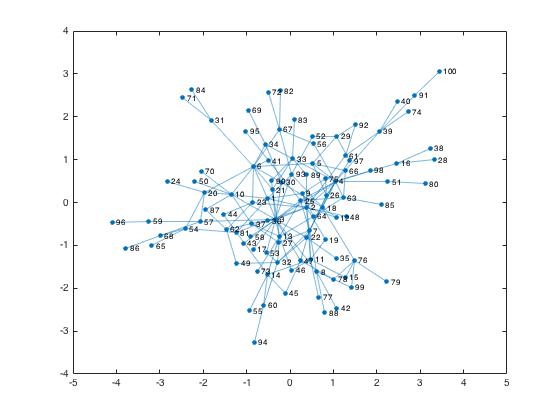}
\caption{The network generated to test the generic implementation of
  the algorithm}\label{fig:fig2}
\end{subfigure}
\end{figure}

\section{Conclusion}

\label{sec:conclusion}

In this paper we have shown how to characterize and compute  the optimal mechanism for a network market. 
We observed in particular that the allocation problem for the optimal and the standard mechanism are the same.
We have proposed an algorithm based on a fixed point to solve the allocation problem and derived regularity properties of the solution.
Our contribution provides a direction to benchmark mechanism proposals.

\bibliographystyle{elsarticle-num}

\begin{thebibliography}{10}
\expandafter\ifx\csname url\endcsname\relax
  \def\url#1{\texttt{#1}}\fi
\expandafter\ifx\csname urlprefix\endcsname\relax\def\urlprefix{URL }\fi
\expandafter\ifx\csname href\endcsname\relax
  \def\href#1#2{#2} \def\path#1{#1}\fi

\bibitem{escobar2010monopolistic}
J.~F. Escobar, A.~Jofr{\'e}, Monopolistic competition in electricity networks
  with resistance losses, Economic theory 44~(1) (2010) 101--121.

\bibitem{escobar2008equilibrium}
J.~F. Escobar, A.~Jofr{\'e}, Equilibrium analysis of electricity auctions,
  Department of Economics Stanford University.

\bibitem{NicolasFigueroaAlejandroJofrBenjaminHeymann}
N.~Figueroa, A.~Jofr{\'e}, B.~Heymann, Cost-minimizing regulations for a
  wholesale electricity market.

\bibitem{myerson1981optimal}
R.~B. Myerson, Optimal auction design, Mathematics of operations research 6~(1)
  (1981) 58--73.

\bibitem{Heymann:2016aa}
B.~Heymann, A.~Jofr{\'e}, Mechanism design and auctions for electricity network
  (2016).

\bibitem{aussel2013electricity}
D.~Aussel, R.~Correa, M.~Marechal, Electricity spot market with transmission
  losses, Management 9~(2) (2013) 275--290.

\bibitem{anderson2013mixed}
E.~J. Anderson, P.~Holmberg, A.~B. Philpott, Mixed strategies in discriminatory
  divisible-good auctions, The RAND Journal of Economics 44~(1) (2013) 1--32.

\bibitem{hu2007using}
X.~Hu, D.~Ralph, Using epecs to model bilevel games in restructured electricity
  markets with locational prices, Operations research 55~(5) (2007) 809--827.

\bibitem{Altman2006}
E.~Altman, T.~Boulogne, R.~El-Azouzi, T.~Jim{\'{e}}nez, L.~Wynter, {A survey on
  networking games in telecommunications}, Computers and Operations Research
  33~(2) (2006) 286--311.
\newblock \href {http://dx.doi.org/10.1016/j.cor.2004.06.005}
  {\path{doi:10.1016/j.cor.2004.06.005}}.

\bibitem{Babaioff2009}
M.~Babaioff, E.~Pavlov, N.~Nisan, {Mechanisms for a Spatially Distributed
  Market}, Games and Economic Behavior (GEB) 66.

\bibitem{Cho2003}
I.-K. Cho, {Competitive Equilibrium in a Radial Network}, The RAND Journal of
  Economics 34~(3) (2003) 438.
\newblock \href {http://dx.doi.org/10.2307/1593740}
  {\path{doi:10.2307/1593740}}.

\bibitem{laffont2009theory}
J.-J. Laffont, D.~Martimort, The theory of incentives: the principal-agent
  model, Princeton university press, 2009.

\bibitem{krishna2009auction}
V.~Krishna, Auction theory, Academic press, 2009.

\bibitem{Roughgarden2016}
T.~Roughgarden, Twenty lectures on algorithmic game theory, Cambridge
  University Press, 2016.

\bibitem{nisan2007introduction}
N.~Nisan, Introduction to mechanism design (for computer scientists),
  Algorithmic game theory 209 (2007) 242.

\bibitem{aumann1992handbook}
R.~J. Aumann, S.~Hart, Handbook of game theory with economic applications,
  Vol.~2, Elsevier, 1992.

\bibitem{topkis1998supermodularity}
D.~M. Topkis, Supermodularity and complementarity, Princeton university press,
  1998.

\bibitem{gibbons1992game}
R.~Gibbons, Game theory for applied economists, Princeton University Press,
  1992.

\bibitem{correa2009planner}
J.~R. Correa, N.~Figueroa, On the planner's loss due to lack of information in
  bayesian mechanism design, in: Algorithmic Game Theory, Springer, 2009, pp.
  72--84.

\bibitem{Philpott2012}
R.~Palma-Benhke, A.~Philpott, A.~Jofr{\'{e}}, M.~Cort{\'{e}}s-Carmona,
  {Modelling network constrained economic dispatch problems}, Optimization and
  Engineering 14~(3) (2013) 417--430.
\newblock \href {http://dx.doi.org/10.1007/s11081-012-9203-5}
  {\path{doi:10.1007/s11081-012-9203-5}}.

\bibitem{Rockafellar2009}
R.~T. Rockafellar, R.~J.-b. Wets, {Variational analysis}.

\bibitem{Rockafellar1993}
R.~T. Rockafellar, {Lagrange Multipliers and Optimality}, SIAM Review 35~(2)
  (1993) 183--238.
\newblock \href {http://dx.doi.org/10.1137/1035044}
  {\path{doi:10.1137/1035044}}.

\bibitem{Berge1997}
C.~Berge, {Topological Spaces: including a treatment of multi-valued functions,
  vector spaces, and convexity}, 1997.

\bibitem{bagnoli2005log}
M.~Bagnoli, T.~Bergstrom, Log-concave probability and its applications,
  Economic theory 26~(2) (2005) 445--469.

\bibitem{laffont1988dynamics}
J.-J. Laffont, J.~Tirole, The dynamics of incentive contracts, Econometrica:
  Journal of the Econometric Society (1988) 1153--1175.

\bibitem{gb08}
M.~Grant, S.~Boyd, Graph implementations for nonsmooth convex programs, in:
  V.~Blondel, S.~Boyd, H.~Kimura (Eds.), Recent Advances in Learning and
  Control, Lecture Notes in Control and Information Sciences, Springer-Verlag
  Limited, 2008, pp. 95--110.

\bibitem{cvx}
M.~Grant, S.~Boyd, {CVX}: Matlab software for disciplined convex programming,
  version 2.1 (Mar. 2014).

\bibitem{barabasi1999emergence}
A.-L. Barab{\'a}si, R.~Albert, Emergence of scaling in random networks, science
  286~(5439) (1999) 509--512.

\bibitem{sundaram1996first}
R.~K. Sundaram, A first course in optimization theory, Cambridge university
  press, 1996.

\end{thebibliography}

\appendix

\section{Proof of Lemma \ref{lemma:H1}}
\label{appendix:lemmaH1}
\begin{proof}
By definition
\begin{eqnarray*}
X(a^1\ldots a^{k-1},b,a^{k+1} \ldots a^N) -
X(a^1\ldots a^{k-1},c,a^{k+1}\ldots a^N)=\\
V(a^1\ldots b\ldots a^N) -
V(a^1\ldots c\ldots a^N)+\\
\sum_{j\neq k}a^j [
Q^j(a^1\ldots b\ldots a^N) -
Q^j(a^1\ldots c\ldots a^N)]\\
+bQ^k(a^1\ldots b\ldots a^N) -
cQ^k(a^1\ldots c\ldots a^N)\\
=\int^c_b Q^k(a^1\ldots s\ldots a^N)\dd s+\sum_{j\neq k} a^j [
Q^j(a^1\ldots b\ldots a^N) -
Q^j(a^1\ldots c\ldots a^N) ]
\\+bQ^k(a^1\ldots b\ldots a^N) -
cQ^k(a^1\ldots c\ldots a^N).
\end{eqnarray*}
We use (H1) for the last equality.
Then we apply a telescopic formula
\begin{eqnarray*}
  X(a) - X(b) = X(a^1 \ldots a^N) - X(b^1, a^2 \ldots a^N)+\\
X(b^1, a^2\ldots  a^N) - X(b^1, b^2 \ldots  a^N)
+\ldots \\
+ X(b^1\ldots  b^{N^1}, a^N) - X(b^1\ldots  b^N)  \\
=\sum_{k=1}^N ( \int_{a^k}^{b^k} Q^k(b^1\ldots s\ldots a^N)\dd s)+\\
\sum_{k=1}^N 
\sum_{j < k}
b^j [ Q^j(b^1\ldots b^{k-1},a^k, a^{k+1} \ldots a^N) -
Q^j(b^1\ldots b^{k-1},b^k, a^{k+1} \ldots a^N)]\\+
\sum_{k=1}^N  \sum_{j >k}
a^j [ Q^j(b^1\ldots b^{k-1},a^k, a^{k+1} \ldots a^N) -
Q^j(b^1\ldots b^{k-1},b^k, a^{k+1} \ldots a^N) ]\\
+\sum_{k=1}^N  a^kQ^k(b^1\ldots b^{k-1}, a^k, a^{k+1}\ldots a^N) -b^kQ^k(b^1\ldots
  b^{k-1} \ldots b^{k}, a^{k+1}\ldots a^N)
\end{eqnarray*}
Reordering the last three terms, we get
\begin{eqnarray*}
\sum_{j=1}^N  \sum_{k > j}
b^j [ Q^j(b^1\ldots b^{k-1},a^k, a^{k+1} \ldots a^N) -
Q^j(b^1\ldots b^{k-1},b^k, a^{k+1} \ldots a^N)]\\+
\sum_{j=1}^N  \sum_{k<j}
a^j [ Q^j(b^1\ldots b^{k-1},a^k, a^{k+1} \ldots a^N) -
Q^j(b^1\ldots b^{k-1},b^k, a^{k+1} \ldots a^N) ]\\
+\sum_{j=1}^N  a^jQ^j(b^1\ldots b^j-1, a^j, a^{j+1}\ldots a^N) -b^jQ^j(b^1\ldots
  b^{j-1} \ldots b^{j}, a^{j+1}\ldots a^N)\\
=
\sum_{j=1}^N  \{ b^j \sum_{k > j}
[ Q^j(b^1\ldots b^{k-1},a^k, a^{k+1} \ldots a^N) -
Q^j(b^1\ldots b^{k-1},b^k, a^{k+1} \ldots a^N)]\\
+  a^jQ^j(b^1\ldots b^{j-1}, a^j, a^{j+1}\ldots a^N) -b^jQ^j(b^1\ldots
  b^{j-1} \ldots b^{j}, a^{j+1}\ldots a^N) +\\
a^j \sum_{k<j}
[ Q^j(b^1\ldots b^{k-1},a^k, a^{k+1} \ldots a^N) -
Q^j(b^1\ldots b^{k-1},b^k, a^{k+1} \ldots a^N) ] 
  \}\\
= \sum_j^N  a^jQ^j(a^1 \ldots a^N) -  b^jQ^j(b^1 \ldots b^N)
\end{eqnarray*}
We end up with 
\begin{equation}
 X(a) - X(b) = \sum_{j=1}^N ( a^j Q^j (a) - b^j Q^j (b) +
 \int_{a^j}^{b^j} Q^j (b^1\ldots b^{j-1},t,a^{j+1} \ldots a^N)\dd t )
\end{equation}
\end{proof}

\section{On $\Set$ and the regularity of $q$}
\label{appendix:qc1}
Remember that the set $\Set$ corresponds to the points of transition  between the two possibilities
defined by the first order condition \eqref{eq:FOC}:
\begin{equation}
\Set = \{ c\in\CC^n, q_i(c) = j\bar{q} \ \mbox{and} \
\lambda_i(c) = c_{j'} \ \mbox{for some} \  i\in I, j\in J,
j'\in \{j,j+1 \}\}.
\end{equation}
Our first goal is to  show
that $\Set$ is a finite  union of sets of zero measure (Lemmas
\ref{lemma:SinS} and \ref{lemma:nullMeas1}).
To do so, we apply the implicit functions theorem. 
From this we deduce the regularity of $q$ (proof of Lemma \ref{lemma:qc1}).
For any $I_A$, $I_B$ partition of $I$, and $I_C \subset I_B$ not
empty, $j\in J^I$ and $j' \in J^I$ such that for all $i$, $j'\in
\{j_i,j_i+1 \}$, we denote  by $S(I_A,I_B,I_C,j,j')$ the set
 \begin{equation}
 \left 
\{ c\in\CC^n \ \mbox{such that for any $i\in I$}
\begin{cases}
i\in I_A \Rightarrow & \lambda_i(c) = c_i^{j_i'} \ \mbox{and} \ q_i(c)
\notin \N\bar{q}\\
i\in I_B \Rightarrow &q_i(c) = j_i \bar{q} \\
i\in I_C \Rightarrow & \lambda_i(c) = c_i^{j'_i}  
\end{cases}
\right \} .
\end{equation}
For an element $c$ of such set, we denote by $M$ the matrix
\begin{equation}
\label{eq:M1}
M(c) = \left(\frac{\partial F_i(\lambda(\cc))}{\partial
  \lambda_j}\right)_{(i,j)\in I_B}.
\end{equation}
We need to study the invertibility of $M$  to apply the implicit
functions theorem (Lemma \ref{lemma:Minvert}).
Note that the function $S$ is defined on a finite set. We use the image
of $S$  to show that
the measure of $\Set$ with respect to the Lebesgue measure is zero.
We first show in the next
lemma that $\Set$ is included in the finite union of
the $S(I_A,I_B,I_C,j,j')$ family. Then we will show that each element
of this family has a measure equal to zero.   
\begin{lemma}
\label{lemma:SinS}
$\Set \subseteq \cup S(I_A,I_B,I_C,j,j')$
\end{lemma}  
\begin{proof}
Take  $c\in \Set$, then by definition of $\Set$, there exist $i\in I$,
$j\in J$ and  $j'\in \{j,j+1 \}$ such that $q_i(c) = j\bar{q}$ and
$\lambda_i(c) = c_{j'}$, therefore $I_C$ is not empty. By Lemma
\ref{lemma:unicity}, for all $i\in I$, $i$ is in $I_A$ or $I_B$. Hence we
have a set $S(I_A,I_B,I_C,j,j')$  such that $c$ is in this set, so $\Set$ is included in the
union of those sets.
\end{proof}
\begin{lemma}
\label{lemma:Minvert}
\label{lemma:M}
For any $c\in \C^n $ the matrix $M(c)$  is invertible.  
\end{lemma}
\begin{proof}
Assume that there are some coefficients
$\alpha_i$ such that $\sum_i\alpha_i M_i=0$ where $M_i$ is the ith
column of $M$. 
Then  by \eqref{eq:der1} and \eqref{eq:der2}, the ith row of this
relation writes:
\begin{equation}
\label{eq:1029}
\alpha_i \sum_{j\in  V(i)}\frac{\lambda^2_j}{r_{i,j}(\lambda_i+\lambda_j)^3}=
\sum_{j\in V(i), j\in I_B} \frac{\alpha_j\lambda_i\lambda_j}{r_{i,j}(\lambda_i+\lambda_j)^3}.
\end{equation}
We denote $b_{i,j} =
\frac{\lambda_j^2\lambda_i}{r_{i,j}(\lambda_i+\lambda_j)^3}$ and $a_i=
\frac{\alpha_i}{\lambda_i}$.
Then \eqref{eq:1029} is equivalent to 
\begin{equation}
a_i = \sum_{j\in V(i), j \in I_B}a_j\frac{b_{i,j}}{\sum_{k\in V(i)}b_{i,k}}
\end{equation}

Insofar as we can slightly perturb the demand, we assume without loss of generality that 
t it is not possible to produce a multiple
of $\bar{q}$ at each node and satisfy exactly the nodal
constraints ($\star$).

Considering the biggest $a_i$,  we get that all $a_i$  are equal by convexity, thus  either all are equal to zero or
\begin{equation}
\sum_{j\in V(i)}b_{i,j}=
\sum_{j\in V(i), j\in I_B} b_{i,j}
\end{equation}
which is not the case since $I_A$ is not empty by ($\star$).
\end{proof}
Next we show that $S(I_A,I_B,I_C,j,j')$ has a zero Lebesgue measure.
\begin{lemma}
\label{lemma:nullMeas1}
For any $I_A$, $I_B$ partition of $I$, and $I_C \subset I_B$ not
empty, $j\in J^I$ and $j' \in J^I$ such that for all $i$, $j'\in
\{j,j+1 \}$, the measure  of the set  $S(I_A,I_B,I_C,j,j')$ is zero. 
\end{lemma}
\begin{proof}
We assume in the market description that it is not possible to
produce a multiple $\bar{q}$  at each node and satisfy exactly the
nodal constraints (($\star$). Therefore it is not possible that $I_B=I$, therefore $I_A$
is not empty. By definition of  $S_{I_A,I_B,I_C,j,j'}$, for   all $i\in I_B$,
 \begin{equation}
F_i(c_{I_A}^{j'},
\lambda_{I_B}(c))=q_i(c) = j_i\bar{q},
\end{equation}
which is a system of  equations in $\lambda_{I_B}$ parametrized
by $c_{I_A}^{j'}$. Let  $c\in\CC$ such that the system is
satisfied,  by Lemma  \ref{lemma:M}, we can apply the implicit function
theorem, hence there is a ball around $c$ in which
$S(I_A,I_B,I_C,j,j')$ is included in a smooth surface. 
By compacity of $\CC$,  we can choose a sequence dense in
$S(I_A,I_B,I_C,j,j')$.
We apply the result to each element of this sequence. By density,
$S(I_A,I_B,I_C,j,j') $ is a countable union of smooth surfaces.
Therefore the measure of $S(I_A,I_B,I_C,j,j')$ is zero.
\end{proof}

A direct consequence of Lemma \ref{lemma:nullMeas1} and Lemma \ref{lemma:SinS} is 
\begin{lemma}
The measure of $\Set$ is zero. 
\end{lemma}

We  proceed with the proof of  Lemma \ref{lemma:qc1}.
\begin{proof}[of lemma \ref{lemma:qc1}]
Let $\cc=(c_1 \ldots c_n) \in \CC^n \backslash \Set $.
Let us show that $q$ is infinitely differentiable at $\cc$.
We consider the two assertions: 
$$A_i = "\exists k_i,\quad F_i(\lambda(\cc)) \in ]k_i-1,k_i[\bar{q} \quad
\mbox{and} \quad \lambda_i = c_i^k"$$
$$B_i = "\exists k_i,\quad F_i(\lambda(\cc)) =k_i\bar{q} \quad
\mbox{and} \quad \lambda_i \in] c_i^k,c_i^{k+1}["$$
By Lemma \ref{lemma:unicity} and by defintion of $\Set$,  for any $i\in
I$ either $A_i$ or $B_i$ is true, but never both. We denote by 
$I_A$ (resp. $I_B$ ) the set of elements of $I$ for which $A_i$ (resp. $I_B$)
is true.
If  $A_i$ is true for all $i$  then there is a neighborhood $V$ of
$\cc$ such that for any element $\tilde{\cc}$ of $V$,  $F_i(\tilde{c})
\in ]k_i-1,k_i[\bar{q}$, therefore on  $V$, $\lambda(\tilde{\cc})= \tilde{\cc}$.

Else $I_B$ is not empty and by definition of $B_i$
\begin{equation}
\forall i\in I_B \quad F_i(\lambda_{I_A}, \lambda_{I_B}) = \bar{q}j_i,
\end{equation}
which we can see as an equation in $ \lambda_{I_B}$ parametrized by
$\lambda_{I_A}$.
This equation is satisfied at $\lambda(\cc)$.
If we denote by $M$ the matrix
\begin{equation}
M = \left(\frac{\partial F_i(\lambda(\cc))}{\partial
  \lambda_j}\right)_{(i,j)\in I_B},
\end{equation}
then $M$ is invertible (see lemma \ref{lemma:M}), 
the implicit function theorem applies and 
 there exists  a function $\lambda_{I_B}$ so that in a neighborhood $V$ of
$\cc$, for all 
$i\in I_B$,  we have  $F_i(\lambda_{I_A}, \lambda_{I_B}(\lambda_{I_A})) = \bar{q}k_i$. 
Moreover, since $F_i$ is $C^\infty$ on $[c_*,c^*]^n$, $\lambda_{I_B}$ is
$C^\infty$ on $V$.
Then if  $\tilde{c}\in V$,  $(\tilde{\cc},
\lambda_{I_B}(\tilde{\cc}))$ checks  the first order condition thus by
uniqueness  $c_{I_A},\lambda_{I_B}(\tilde{c})$ is the dual solution, and 
so, $q_i=F_i(\lambda_{I_B}(\tilde{\cc}), \tilde{\cc})$ 
for all $i\in I$  on $V$, so $q_i$ is
$C^\infty$ at $\cc$.
This concludes the proof of the first  part of the lemma.

The continuity of $q$ comes from Berge maximum principle (see Theorem
9.17 in \cite{sundaram1996first}) in a convex setting.
\end{proof}

The next lemma is an important component for the proof of Theorem \ref{th:Qc1}.
\begin{lemma}
\label{lemma:nullMeasure2}
Let $i \in I$ and $c_i\in C_i$, then the Lebesgue measure of the  set 
\begin{equation}
\Set_i(c_i) =\{c_{-i} \in C_{-i}, (c_i,c_{-i}) \in \Set \}
 \end{equation}
is zero.
\end{lemma}
\begin{proof}
Using Lemma \ref{lemma:SinS}, 
$\Set_i(c_i) \subseteq \{c_{-i} \in C_{-i}, (c_i,c_{-i}) \in  \cup
S(I_A,I_B,I_C,j,j') \}$. Now let $c_{-i}\in \Set_i(c_i)$,  $I_A$, $I_B$ a partition of $I$, and
$I_C\subseteq I_B$ not empty, and $j$, $j'$ such that 
$(c_i,c_{-i})\in S(I_A,I_B,I_C,j,j') $.
There are three possible cases: 
\begin{itemize}
\item $i \in I_A$  then as explained in the proof of Lemma
  \ref{lemma:nullMeas1}, $S(I_A,I_B,I_C,j,j')$ is locally a surface
  parametrized by $c_{i}$ so by projection over an hyperplane of the type $c_i=x$ it also a surface in $\CC_{-i}$.
\item $i \in I_B\backslash I_C$ locally, $q$ is independant of $c_i$
  therefore if $S(I_A,I_B,I_C,j,j')\cap (c_i,\Set_i(c_i))$ is of strictly
  positive measure, then
  $S(I_A,I_B,I_C,j,j')$ has also a strictly positive measure in
  $\CC^n$, since this is not true, $S(I_A,I_B,I_C,j,j')\cap
  (c_i,\Set_i(c_i))$ is of zero measure in the neighborhood.
\item Else $i \in  I_C$, which  is the tricky part.
First by definition of $I_C$, for any element $c$ of   $S(I_A,I_B,I_C,j,j')$, $q_i(c)
= j_i\bar{q}$ and $\lambda_i (c) = c_i^{j_i'}$.
Without loss of generality, we assume $j_i' = j_i$, the other case can
be treated similarly.
Then we make the observation that we do not modify the $c_{-i}$ of 
$S(I_A,I_B,I_C,j,j')$ if we set  $c_i^{j+1}= c_i^{j}$. Since we are
interested in $S(I_A,I_B,I_C,j,j')\cap (c_i,\Set_i(c_i))$, we can
assume without loss of generality that $c_i^{j+1}= c_i^{j}$.
Then we have reduced to the case  $i \in I_A$. 
\end{itemize}
We conclude as in the proof of Lemma \ref{lemma:nullMeas1}.
\end{proof}

\end{document}